\documentclass[11pt,a4paper]{article}
\pdfoutput=1
\usepackage[utf8]{inputenc}
\usepackage[T1]{fontenc}

\usepackage[dvipsnames,usenames]{color}
\usepackage[colorlinks=true,urlcolor=Blue,citecolor=Green,linkcolor=BrickRed]{hyperref}
\usepackage[usenames,dvipsnames]{xcolor}
\usepackage{amsthm,amsmath,amssymb, mathabx}
\usepackage{geometry}
\allowdisplaybreaks

\usepackage[ruled,noend,linesnumbered]{algorithm2e}     
\usepackage{enumerate,comment}
\usepackage{url}
\usepackage[capitalise]{cleveref}
\usepackage{todonotes}
\usepackage{standalone}
\usepackage[noadjust]{cite}
\usepackage{mathtools}
\usepackage[inline]{enumitem}

\usepackage{graphicx, subfigure}
\usepackage{array}
\usepackage{xspace}
\usepackage{graphicx}
\usepackage{float}
\usepackage{amsfonts}

\usepackage{booktabs}
\usepackage{multirow}
\usepackage{tablefootnote}
\usepackage{threeparttable}
\newcommand{\ra}[1]{\renewcommand{\arraystretch}{#1}}

\makeatletter
\newcommand*{\bdiv}{%
  \nonscript\mskip-\medmuskip\mkern5mu%
  \mathbin{\operator@font div}\penalty900\mkern5mu%
  \nonscript\mskip-\medmuskip
}
\makeatother

\newcommand{\ceil}[1]{\left\lceil #1 \right\rceil}
\newcommand{\floor}[1]{\left\lfloor #1 \right\rfloor}

\newcommand{\set}[1]{\{ #1 \}}

\newcommand{\Oh}{{O}}
\newcommand{\common}{\mathsf{common}}

\newcommand{\id}{\mathsf{id}}
\newcommand{\lightdepth}{\mathsf{lightdepth}}

\newcommand{{\appx}}[2][2]{\lfloor{#2}\rfloor _{#1}}
\newcommand{\pre}{\mathsf{pre}}
\newcommand{\lightrange}[1]{\mathsf{L}_{#1}}
\newcommand{\heavy}{\mathsf{heavy}}

\newcommand{\size}[1]{\left| #1 \right|}
\newcommand{\NCA}{\mathsf{NCA}}
\newcommand{\NCSA}{\mathsf{NCSA}}
\newcommand{\NCH}{\mathsf{NCH}}
\newcommand{\height}{\mathsf{height}}
\newcommand{\distance}{\mathsf{d}}
\newcommand{\nil}{\mathsf{nil}}

\newcommand{\accum}{\mathsf{a}}
\newcommand{\eps}{\varepsilon}
\newcommand{\htree}{(h,M)\text{-tree}}

\newcommand{\lightcount}{\lightdepth}
\newcommand{\rootdist}{\mathsf{root}\text{-}\mathsf{distance}}
\newcommand{\treeroot}{\mathsf{root}}
\newcommand{\collapsed}{\mathcal{C}}
\newcommand{\hphead}{\mathsf{head}}

\newtheorem{theorem}{Theorem}[section]
\newtheorem{property}[theorem]{Property}

\newtheorem{lemma}[theorem]{Lemma}
\newtheorem{observation}[theorem]{Observation}

\theoremstyle{definition}   

\usepackage{authblk}
\theoremstyle{remark}

\graphicspath{{./figures/}}

\begin{document}

\title{Optimal Distance Labeling Schemes for Trees}
\author[1]{Ofer Freedman\thanks{The research was supported in part by Israel Science Foundation grant 794/13.}}
\author[1]{Pawe\l{} Gawrychowski\protect\footnotemark[1]}
\author[2]{Patrick K. Nicholson}
\author[1]{Oren Weimann\protect\footnotemark[1]}
\affil[1]{University of Haifa, Israel}
\affil[2]{Bell Labs, Dublin, Ireland}

\date{}
\maketitle

\begin{abstract}

Labeling schemes seek to assign a short label to each node in a network, so that a function on two nodes (such as distance or adjacency) can be computed by examining their labels alone. For the particular case of trees, following a long line of research, optimal bounds (up to low order terms) were recently obtained for adjacency labeling [FOCS '15], nearest common ancestor labeling [SODA '14], and ancestry labeling [SICOMP '06]. In this paper we obtain optimal bounds for distance labeling. We present labels of size $1/4\log^2n+o(\log^2n)$, matching (up to low order terms) the recent $1/4\log^2n-\Oh(\log n)$ lower bound [ICALP '16]. 

Prior to our work, all distance labeling schemes for trees could be reinterpreted as {\em universal trees}. A tree $T$ is said to be universal if any tree on $n$ nodes can be found as a subtree of $T$. 
A universal tree with $|T|$ nodes implies a  distance labeling scheme with label size $\log |T|$. In 1981, Chung et al. proved that any distance labeling scheme based on universal trees requires labels of size $1/2\log^2 n -\log n \cdot \log\log n+\Oh(\log n)$. Our scheme is the first to break this lower bound, showing a separation between distance labeling and universal trees.\\

The $\Theta (\log^2 n)$ barrier for distance labeling in trees has led researchers to consider distances bounded by $k$. The size of such labels  was improved from $\log n+\Oh(k\sqrt{\log n})$ [WADS '01] to $\log n+\Oh(k^2(\log(k\log n))$  [SODA '03] and finally to $\log n+\Oh(k\log(k\log(n/k)))$ [PODC '07].
We show how to construct labels whose size is the minimum between $\log n+\Oh(k\log((\log n)/k))$ and $\Oh(\log n \cdot \log(k/\log n))$. We complement this with almost tight lower bounds of $\log n+\Omega(k\log(\log n / (k\log k)))$ and
$\Omega(\log n \cdot \log(k/\log n))$. Finally, we consider $(1+\eps)$-approximate distances. We show that the recent labeling scheme of [ICALP '16] can be easily modified to obtain an $\Oh(\log(1/\eps)\cdot \log n)$ upper bound and we prove a matching $\Omega(\log(1/\eps)\cdot \log n)$ lower bound. \end{abstract}

\thispagestyle{empty}
\clearpage
\setcounter{page}{1}

\section{Introduction}
Labeling schemes seek to assign a short label to each vertex in a network, so that a function on two nodes (such as distance or adjacency) can be computed by examining their labels alone. This is particularly desirable in distributed settings, where nodes are often processed using only some locally stored data. 
Recently, with the rise in popularity of distributed computing platforms such as Spark 
 and Hadoop,  
  labeling schemes have found renewed interest. 
Indeed, the goal of minimizing the size of the maximal label has been the subject of a great deal of recent research~\cite{alstrup2005labeling,abiteboul2006compact,fischer2009short,fraigniaud2010compact,alstrup2014near,alstrup2015adjacency,alstrup2015optimal,petersen2015near,alstrup2015distance,alstrup2016simpler}.
For the particular case of trees, the functions that have been studied are distance~\cite{peleg2000proximity,gavoille2004distance,alstrup2005labeling,alstrup2015distance,gavoille2007distributed}, adjacency~\cite{alstrup2015optimal,alstrup2002small,bonichon2007short}, nearest common ancestor~\cite{fischer2009short,alstrup2014near}, and ancestry~\cite{abiteboul2006compact,fraigniaud2010compact} (a recent survey of these results can be found here~\cite{rotbart2016new}). Tree labeling schemes have recently found new uses in large scale graph processing. For example, distance oracles for general graphs use distance labelings for spanning trees rooted at judiciously chosen vertices~\cite{akiba2012treewidth,akiba2013distance,ajwani2015oracle}.

\paragraph{Universal trees.}

A particularly clean way of looking at labeling schemes is through \emph{universal graphs}. A graph $G$ is said to be {\em universal} for a given family of graphs, if every graph in the family is an induced subgraph of $G$. Similarly, a tree $T$ is said to be universal for {\em all} trees on $n$ nodes if any tree on $n$ nodes can be found as a subtree of $T$. 
For {\em adjacency labeling in graphs}, Kannan et al.~\cite{Kannan} observed that if a family of graphs has a universal graph with  $|G|$ vertices then it has an adjacency labeling scheme with label size $\log |G|$, and vice versa. 
For {\em distance labeling in trees}, until the present work, this statement was only known to be true in one direction. Namely, a universal tree $T$ of all trees on $n$ nodes implies a distance labeling scheme with label size $\log |T|$. We prove that the converse is in fact not true.  

The use of universal trees is powerful, but it is limited. Already 50 years ago, Goldberg and Livshits~\cite{gol1968minimal} showed how to construct a universal tree $T$ that is of size $|T| = n^{(\log n - 2\log \log n + \Oh(1))/2}$ which was shown by Chung et al.~\cite{graham1981trees} to be the smallest possible up to the $\Oh(1)$ error term. This shows the first limitation of using universal trees for distance labeling: there is a lower bound of $\log |T| = 1/2\log^2 n -\log n \cdot \log\log n+\Oh(\log n)$ on the label size. The second limitation is the query time. The universal tree construction of Goldberg and Livshits was given before labeling schemes were ever invented. Of course, one could naively use their universal tree $T$ for distance labeling of an arbitrary tree on $n$ nodes by finding its isomorphic subtree in $T$ and assigning labels which are just the IDs of the nodes in $T$. However, such a non-algorithmic labeling would require prohibitive query time and space since $T$ needs to be computed. This latter limitation was overcome by algorithmic labeling schemes achieving logarithmic query time: An upper bound of $\Oh(\log^2 n)$ bits on the label size was first shown by Peleg~\cite{peleg2000proximity} and a lower bound of  $1/8 \log^2n - \Oh(\log n)$ bits was shown by Gavoille et al.~\cite{gavoille2004distance}. Very recently, Alstrup et al.~\cite{alstrup2015distance}  improved the lower bound to  $1 / 4 \log^2 n - \Oh(\log n)$ and  observed that the upper bound can be improved to $1/2 \log^2 n + \Oh(\log n)$ with a somewhat straightforward use of a nearest common ancestor labeling scheme.

All the above labeling schemes can be reinterpreted as building a universal tree, and are therefore subject to the  $1/2\log^2 n -\log n \cdot \log\log n+\Oh(\log n)$ lower bound of Chung et al. In other words, the scheme of Alstrup et al. is optimal (up to low order terms) amongst all schemes that translate to universal trees. To see why the scheme of Alstrup et al. indeed translates to a universal tree, we show in~\cref{sec:lvl-anc} that their scheme can be casted as a level-ancestry scheme and we show in~\cref{sec:lb} that every level-ancestry scheme translates to a universal tree.

We give the first distance labeling scheme that does not translate to a universal tree. This enables us to circumvent the Chung et al.~\cite{graham1981trees} lower bound for labels based on universal trees and to match the general lower bound of Alstrup et al.~\cite{alstrup2015distance}. Namely, in~\cref{sec:ub} we prove the following:

\begin{theorem}
\label{thm:ub-tree-dist}
There is a scheme for tree distance labeling with $1/4 \log^2 n + o(\log^2 n)$ bit labels and constant query time.
\end{theorem}

The above theorem means that universal trees capture more than is required for distance labeling. To illustrate this, we need to describe the related problem of {\em level-ancestor labelings}.  

\paragraph{Labeling schemes for level-ancestors.}
In this problem, we are given a rooted tree and seek to assign labels so that we can compute (the label of) any $k$-th ancestor of a node from its label alone. Notice that here a query receives a single label and a value $k$, and 
that all labels must be distinct (no scheme which uses the same label twice can be correct). 

It is not hard to see that labels supporting level-ancestor queries can be used to answer distance queries. Thus, any lower bound for tree distance labeling immediately applies to level-ancestor labeling, but the converse is not true. Nevertheless, it turns out that all previous distance labeling schemes are also level-ancestor schemes. 
Like the labeling scheme of Alstrup et al.~\cite{alstrup2015distance}, our scheme is also based on a \emph{heavy path decomposition} of the tree, which can be seen as a way of transforming an arbitrary tree into an edge-weighted tree of logarithmic depth. 
However, while the labels in~\cite{alstrup2015distance} store the weights of every edge on the path to the root (thus allowing for level-ancestor queries), we show that it is possible to carefully distribute the bits between the labels so that the distance can be computed given any pair of labels, yet a single label is not enough to extract the level-ancestors.

We determine this separation between tree distance labeling and level-ancestor labeling by proving that labeling for distances is roughly half as expensive as labeling for level-ancestors:

\begin{theorem}
\label{thm:lb-lvl-anc}
Any scheme for level-ancestor labeling must use at least $1/2 \log^2 n - \log n \log \log n$ bits for the maximum length label.
\end{theorem}

 We prove the above theorem in~\cref{sec:lb} by showing that, as opposed to distance labeling, no level-ancestor labeling scheme can do 
better than the one based on universal trees. Namely, we prove that any level-ancestor labeling scheme with labels of length $L$ implies a universal rooted tree of size $\Oh(2^{L})$, and then invoke the known lower bound for universal trees~\cite{graham1981trees,gol1968minimal}.
In particular, it means that for level-ancestor queries, the scheme of Alstrup et al.~\cite{alstrup2015distance} is optimal (after some modifications described
in~\cref{sec:lvl-anc}).

\paragraph{Labeling schemes for bounded distances.}
The $\Theta(\log^2 n)$ barrier on distance labeling in trees has initiated a line of research that improves the label size when the distances are bounded: In {\em $k$-distance labeling}, we are given the labels of $u$ and $v$ and need to decide if the length of the $u$-to-$v$ path is at most $k$, and if so return it. 
For $k=1$, this is exactly adjacency labeling,
which was recently shown by Alstrup et al.~\cite{alstrup2015optimal} to require only $\log n+\Oh(1)$
bits. 
For $k\ge 2$, this was first considered by Kaplan and Milo~\cite{kaplan2001short} who showed how to construct labels of length $\log n+\Oh(k\sqrt{\log n})$. The query time was
not explicitly specified in their implementation, but appears to be $\Oh(k)$.
A shorter label of $\log n+\Oh(k^2\log(k\log n))$ bits  was then given by Alstrup, Bille, and Rauhe~\cite{alstrup2005labeling} who also proved that any scheme for $k\geq 2$ (i.e., the scheme is able to answer ``ancestor or sibling''
queries) requires $\log n+\Omega(\log\log n)$ bits. 
Hence the $\Oh(\log\log n)$ addend
cannot be avoided, but it remained unclear what should be the exact dependency on $k$ nor the query time (Alstrup, Bille, and Rauhe considered constant $k$ in which case their bounds are tight and their $O(k^2)$ query time is constant). The labeling scheme of Alstrup, Bille, and Rauhe was then improved by Gavoille and Labourel~\cite{gavoille2007distributed} who presented
a  bound of $\log n+\Oh(k\log(k\log(n/k)))$ bits and $\Oh(k)$ query time solution.

In~\cref{sec:k} we show how to construct a labeling scheme with improved label size and constant query time, and prove an almost matching lower bound. Formally, we prove:

\begin{theorem}
\label{thm:krestricted}
For $k < \log n$, there is a $k$-distance labeling scheme with labels of length
$\log n+\Oh(k\log((\log n)/k))$ bits, and any such scheme requires
$\log n+\Omega(k\log(\log n/(k \log k))$ bits.\\
For $k \ge \log n$, there is a $k$-distance labeling scheme with labels of length
$\Oh(\log n\cdot \log(k/\log n))$ bits, and any such scheme requires
$\Omega(\log n\cdot \log(k/\log n))$ bits.
In both cases, the query time is constant.
\end{theorem}

For the upper bound, our starting point is the scheme of Alstrup, Bille, and Rauhe~\cite{alstrup2005labeling}. 
We observe that, instead of storing the same information for each of the nearest $k$ heavy paths above a node, it is possible to store all information for the topmost of these heavy paths and less information for all the rest. To improve the query
time, we show that only a subtle change is needed in the definition of the so-called
significant preorder numbers. The new definition retains all the nice properties of the previous while being much easier to operate on. Our constant query time  assumes the standard word-RAM model with word size $\Omega(\log n)$.

For the lower bounds we take two different approaches. For $k < \log n$, we show how to construct a
family of trees such that, 
in any $k$-distance labeling scheme, different trees can share some labels but every tree has to introduce many additional unique labels.
For $k \ge \log n$, we use the clever lower bound technique from (unbounded) distance labelings, that was introduced by Gavoille et al.~\cite{gavoille2004distance} and refined by Alstrup et al.~\cite{alstrup2015distance}. It is based on constructing a \emph{weighted} almost complete binary tree, where all the leaves are at the same distance from the root.  After arguing that the labels of nodes in such a tree must be long, the weights are removed by subdividing edges while not increasing the size of the tree by too much.
We show that only a small tweak is required to this known lower bound for distance labeling in order to get a lower bound for $k$-distance labeling.

\paragraph{Labeling schemes for approximate distances.}
Finally, we consider {\em $(1+\eps)$-approximate distance labeling}, where given the labels of $u$
and $v$ we need to output a value that is at least $\distance(u,v)$ and at most 
$(1+\eps)\cdot\distance(u,v)$. For the case $\eps \in [1/\log n, 1)$, Gavoille et
al.~\cite{GavoilleKKPP01}
proved a tight bound of $\Theta(\log(1/\eps)\cdot\log n)$. Very recently,
Alstrup et al.~\cite{alstrup2015distance} considered the general trade-off and designed,
for any constant $\eps \leq 1$, an $\Oh(\log n)$ bit labeling
scheme. In~\cref{sec:approx} we show that their solution can be easily made to produce labels of size  $\Oh(\log(1/\eps)\cdot \log n)$ and that this is the best possible:

\begin{theorem}
\label{thm:approximate}
For any $\eps \le 1$, there is a $(1+\eps)$-approximate distance labeling scheme with  labels of length $\Oh(\log(1/\eps)\cdot \log n)$, and any such scheme requires
$\Omega(\log(1/\eps)\cdot \log n)$ bits.
\end{theorem}

The lower bound is obtained by reducing  exact distance labeling to $(1+\eps)$-approximate distance labeling.  
This is achieved by appropriately stretching the lengths of the edges in the lower bound instances of Gavoille et al.~\cite{gavoille2004distance}. 
For the upper bound, we slightly modify the scheme of Alstrup et al.~\cite{alstrup2015distance}, which 
originally stored a sequence of integers using simple unary encoding. Such an encoding requires $\Oh(1/\eps \cdot \log n)$ bits. We show that with a more complicated binary encoding
we can obtain a scheme with $\Oh(\log(1/\eps)\cdot \log n)$ bits and a constant query time. 

\noindent
We conclude this section with the following table summarizing our contribution.

\begin{table}[H]
\centering
\begin{threeparttable}
\ra{1.5}
	\begin{tabular}{@{}l l l l@{}}
		\toprule
  		\multicolumn{2}{@{}l}{\bf Label type} & \bf Upper bound & \bf Lower bound\\ 
		\midrule
  		\multicolumn{2}{@{}l}{Exact} & $1/4\log^2 n  + o(\log^2 n)$ & $1/4\log^2 n - O(\log n)$ \cite{alstrup2015distance}\\
  		\multicolumn{2}{@{}l}{Approximate} & $\Oh(\log(1/\eps)\cdot \log n)$ & $\Omega(\log(1/\eps)\cdot \log n)$\\
  		\multirow{2}{*}{$k$-distance} & $k \geq \log n$ & $O(\log n\cdot\log \frac {k} {\log n})$ & $\Omega(\log n\cdot\log \frac{k}{\log n})$ \\
		& $k < \log n$ & $\log n + O(k\log \frac {\log n} {k})$ & $\log n + \Omega(k\log \frac {\log n} {k\log k})$\tnote{1} \\
\bottomrule
	\end{tabular}
\begin{tablenotes}\footnotesize
\item[1] This lower bound only holds for $k = o(\frac{\log n}{\log\log n})$.
\end{tablenotes}
\end{threeparttable}
\end{table}

\section{Preliminaries}
\label{sec:prelim}

We consider a rooted tree $T$, or we arbitrarily root it. We denote the root by $\treeroot(T)$, and the distance between node $v$ to $\treeroot(T)$ by $\rootdist(v)$. We denote the subtree
rooted at $u$ as $T_u$, and the number of nodes of $T$ by $\size{T}$, or simply $n$ if $T$ is known from the context. For two nodes $u,v$, we denote their distance by 
$\distance(u,v)$, their nearest common ancestor by $\NCA(u,v)$.
First, observe that:
$$\distance(u,v) =  \rootdist(u) + \rootdist(v) - 2\cdot\rootdist(\NCA(u,v)).$$ 
This means that a labeling scheme for $\rootdist(\NCA(u,v))$ that assumes $u\ne v$ can actually be used for $\distance(u,v)$ queries  with only $O(\log n)$ additional bits to the label size (the additional bits are simply the distance to the root).
Next, observe that although our input tree is unweighted (i.e., all edges have weight 1), if our distance labeling scheme can handle edges-weights in $\{0,1\}$ then we can assume the input tree is binary and 
that the queries are on leaves only. This can be achieved by connecting every internal node $u$ to a leaf node $u_\ell$ with an edge of weight 0, and then standardly binarizing the tree (by inserting $\Oh(n)$ intermediate nodes with edge-weights 0 connecting them).

\paragraph{Heavy path decompositions.}
We apply a variant of heavy path
decompositions~\cite{thorup2001compact}. We start at the root of the tree $T$ and repeatedly
descend from the current node to its (unique) child whose subtree is of size at least $|T|/2$ as long as possible, that is, we terminate when there is no such child.
Note that this is different than the more common versions in which we descend from the current node $u$ to its child $v$ with the largest subtree until
(depending on the version) we reach a leaf or $|T_v| < |T_u|/2$. 
This gives us a \emph{heavy path} $P$ starting at $\treeroot(T)$ and many subtrees
hanging off the heavy path. We call the edges of $P$ \emph{heavy}, and all other 
edges outgoing from the nodes of $P$ \emph{light}. The construction is then applied
recursively to all subtrees hanging off the heavy path. In the end, each node $u\in T$ has at most
one heavy child, denoted $\heavy(u)$, so we obtain a decomposition of $T$ into disjoint heavy paths (some of which consist of a single node).
The light depth of a node $u\in T$, denoted $\lightdepth(u)$, is the number of light edges on the
path from $u$ to the root and is at most $\log n$~\cite{sleator1983data}.
We order the children of every node $u$ so that $\heavy(u)$ is the rightmost child and
assign preorder numbers $\pre(u)$ to every node $u$. Then, for any node $v\in T_u$, we have
that $\pre(v) \in [\pre(u), \pre(u)+\size{T_u})$.

\begin{figure}[h]
\centering
\subfigure{\includegraphics[width=79mm]{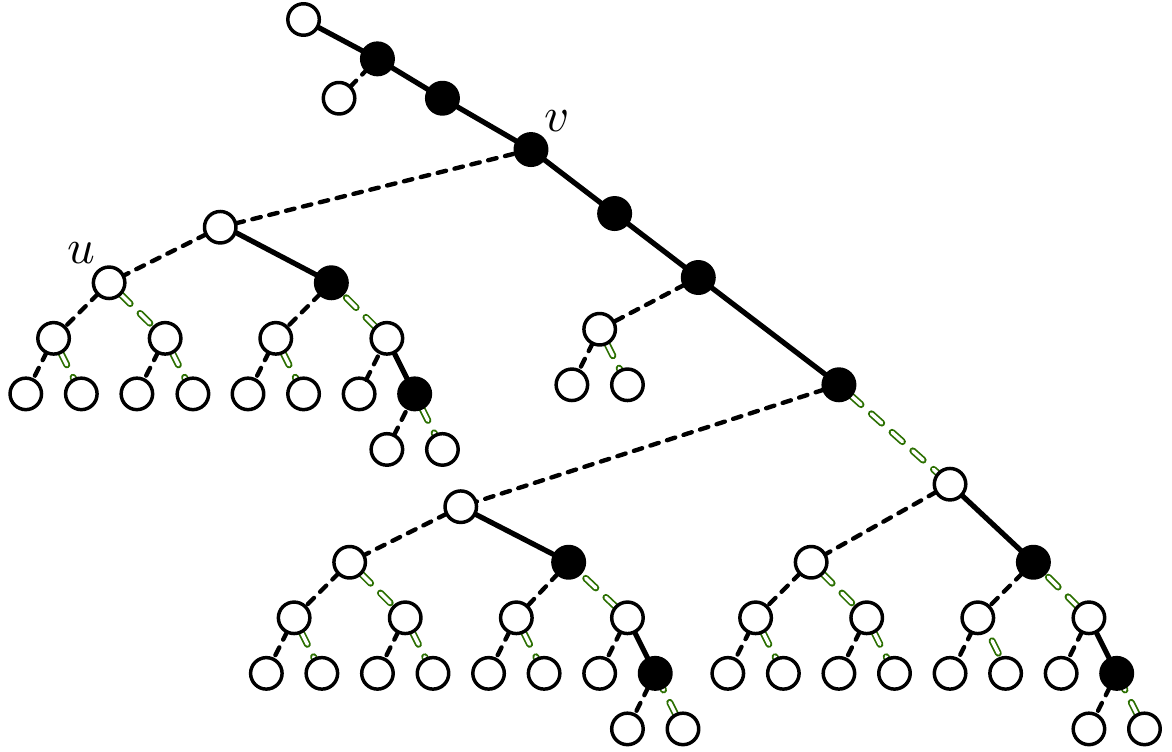}}
\subfigure{\includegraphics[width=79mm]{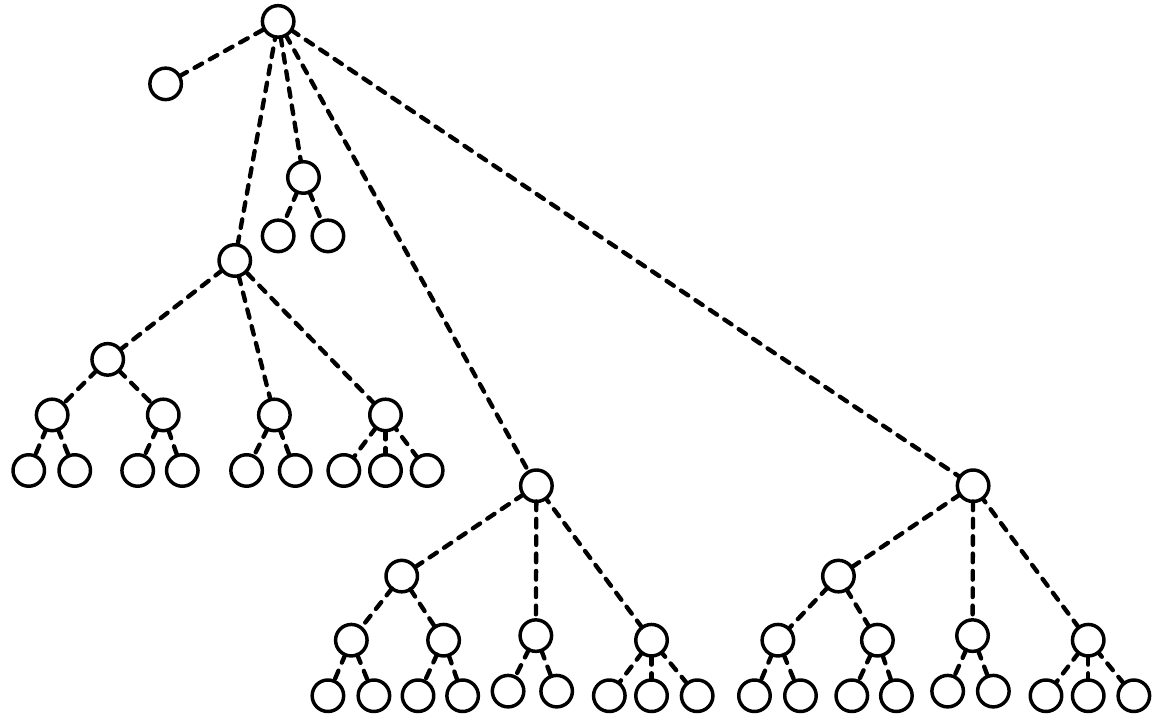}}

\caption{\label{fig:heavy-path-decomp} On the left, a heavy path decomposition of a binary tree $T$.  Light nodes are white, and heavy nodes are black. The heavy edges are solid, the light edges are dashed, and the exceptional edges are dashed and hollowed. On the right, the collapsed tree $\collapsed(T)$.}
\end{figure}

\paragraph{The collapsed tree.}
Given the heavy path decomposition of a {\em binary} tree $T$, we define
its \emph{collapsed tree}, denoted $\collapsed(T)$, whose nodes correspond to heavy paths in $T$. The heavy path starting at $\treeroot(T)$ corresponds to the root of the collapsed tree. Every light edge hanging off this heavy path corresponds to an edge outgoing
from the root of $\collapsed(T)$, and so on.
The children of every node in $\collapsed(T)$ are ordered according to the top-to-bottom order on the hanging subtrees (i.e, if two subtrees connect to the same heavy path $P$ then the one connecting at a lower
depth is to the left of the other). 
Since $T$ is binary, ties can only happen at the last node of the heavy path $P$, in which case we set the right subtree to be the subtree of maximum size, and call the light edge branching to the right subtree  the \emph{exceptional} edge associated with heavy path $P$. See~\cref{fig:heavy-path-decomp} (right). Note that the height of the collapsed tree is at most $\log n$.

Every heavy path $P$ in $T$ is associated with a node $u'$ in $\collapsed(T)$ and every node $u \in P$ is said to be \emph{associated} with $u'$.
We refer to the node $u \in P$ closest to the root of $T$ as the \emph{head} of $P$ and denote it as $\hphead(P)$ or $\hphead(u')$. We use $\lightcount(u,v)$ to denote $\lightdepth(\NCA(u,v))$. Finally, we say that $u$ \emph{dominates} $v$ if the inorder number of $u$'s associated node in $\collapsed(T)$ is smaller than that of $v$'s. 
Observe that (1) If the $\NCA(u,v)$-to-$u$ path in $T$ starts with a light edge and the $\NCA(u,v)$-to-$v$ path starts with a heavy edge then $u$ dominates $v$, and (2) If both these paths start with a light edge then the dominated vertex is the one whose path starts with the exceptional edge. 

\paragraph{Labeling schemes for NCA.}
A nearest common ancestor scheme assigns a unique label to every node, so that given
the labels of nodes $u,v$ we can return the label of $\NCA(u,v)$. Alstrup et al.~\cite{alstrup2014near}
design such a scheme with labels of length $\Oh(\log n)$ bits.
They use a heavy path decomposition that slightly differs from ours, but it can be easily verified that the following lemma still holds:

\begin{lemma}[\cite{alstrup2014near, alstrup2015distance}] \label{lem:NCA scheme} 
There is an $\NCA$ labeling scheme with label size $\Oh(\log n)$, which given the labels of
$u$ and $v$ returns the label of $\NCA(u,v)$ as well as $\lightdepth(u,v)$ in constant time.
\end{lemma}

\paragraph{Encoding integers.} 
To store a single integer $x$, we use Elias $\delta$ codes~\cite{elias1975universal} that require $\log x + \Oh(\log \log x)$ bits. This encoding is {\em self-delimiting}, meaning we can concatenate multiple variable-length values into a single label in a way that each individual value can be decoded later. 
To store a monotone sequence of integers we use the following:

\begin{lemma}
\label{lem:encodingsequence}
A monotone sequence of $s$ integers in $[0,M]$ can be encoded
with $\Oh(s\cdot \max\set{1,\log\frac{M}{s}})$ bits, so that we can:
\begin{enumerate}[label=(\arabic*)]
\item extract the $k^\text{th}$ number in the sequence,
\item find the position of the successor of a given integer in the sequence,
\item given the representation of two sequences, find the longest common suffix of two specified prefixes.
\end{enumerate}
The first operation takes constant time, and the second and third take constant time
if both $s$ and $M$ are $\Oh(\log n)$.
\end{lemma}

\begin{proof}
Let the sequence be $0\leq x_1 \leq x_2 \leq \ldots \leq x_s \leq M$. The encoding consists of
$x_1$ and the differences $x_2-x_1, x_3-x_2,\ldots,x_s-x_{s-1}$. Each number is
encoded using the Elias $\gamma$ code, so the total size of the encoding becomes
$\Oh(s+\sum_{i=1}^s \log (x_i-x_{i-1}))$, where $x_0=0$. By Jensen's inequality, this
is maximized when all numbers are equal, so the total size of the encoding
is $L=\Oh(s\cdot \max\set{1,\log\frac{M}{s}})$.

To provide constant time access to every $x_i$, we need to store some auxiliary data.
We partition the universe $[0,M]$ into blocks of length $b=\frac{M}{s}$. 
For each $x_{i}$, we store $x_{i}\bmod b$. This is done by reserving $\lceil\log b\rceil$
bits for every $i=1,2,\ldots,s$ and arranging them one after another.
We also store $b$ encoded using the Elias $\gamma$ code, so that
in constant time we can calculate where the $\lceil\log b\rceil$ bits storing
$x_{i}\bmod b$ are. This takes $\Oh(\log b+s+s\log b)=\Oh(L)$ space
so far. It remains to show how to encode $y_{i}=x_{i}\bdiv b$.
Notice that $0\leq y_{1}\leq y_{2}\leq\ldots\leq y_{s}\leq s$, so
this is a monotone sequence of $s$ integers from $[0,s]$. We encode
it with a single bit vector of length at most $2s$, which is the concatenation
of $0^{y_{i}-y_{i-1}}1$ for $i=1,2,\ldots,s$ (and $y_{0}=0$). Then,
to extract $y_{i}$ we need to find the position $p$ of the $i^\text{th}$ bit set
to 1 in the bit vector and then return $p-i+1$. By augmenting the bit vector with a select
structure of Clark~\cite[Chapter~2.2]{Clark}, which takes $o(s)$ additional bits of space,
we can retrieve the $i^\text{th}$ bit set to 1 in constant time.
Thus, in $O(L)$ additional space we can encode $x_{i}\bmod b$ and $x_{i}\bdiv b$,
and then recover $x_{i}$ in constant time.

To provide constant time successor queries (when both $s$ and $M$ are $\Oh(\log n)$), we remove all duplicates and store
the resulting sequence $y_1 < y_2 < \ldots < y_r$ in an additional predecessor structure
from the second branch of P\v{a}tra\c{s}cu and Thorup~\cite{PatrascuPredecessor}.
This structure uses $\Oh(r \cdot \log M)$ bits and answers queries in
$\Oh(\log\frac{\log M}{\log \log n})=\Oh(1)$ time. The space can be actually improved
to $\Oh(r \cdot \log \frac{M}{r}))=\Oh(s\cdot \max\set{1,\log\frac{M}{s}})$ as explained in detail by
Belazzougui and Navarro~\cite{DjamalSequences}.

Finally, to compute the longest common suffixes of two specified prefixes given
the encodings of $x_1 \leq x_2 \leq \ldots \leq x_s$ and $y_1 \leq y_2 \leq \ldots \leq y_s$,
we observe that for $s,M=\Oh(\log n)$ the encodings fit in a constant number of machine words.
Hence, we can first shift both encodings (in constant time) to reduce the problem to computing
the longest common suffix. First, we check if $x_s=y_s$. If not, we are done.
Otherwise, we only need to find the longest common suffix of the sequences
of differences. This can be done by first calculating the longest common suffix
of their encodings, and then counting how many differences have their encodings
fully in the common suffix. The former can be done in constant time using
the standard word-RAM operations. The latter can be done by storing an additional
bit vector of length $L$, where we mark the starting position of the encoding of each
$x_{i}$ with a bit set to 1. The bit vector is augmented with the rank structure of Jacobson~\cite{Jacobson},
which takes $o(L)$ additional bits and allows us to count bits set to 1 in any prefix in constant time.
\end{proof}

\paragraph{\boldmath$(h,M)$-trees.} To obtain a lower bound for  distance labeling, Gavoille et
al.~\cite{gavoille2004distance} consider a family of rooted binary trees called $(h,M)$-trees. The trees are weighted
and the weight of every edge is in $[0,M]$. For $h=0$ the tree is a single node. For $h\geq 1$,
the tree consists of a root connected to its single child with an edge of length $M-x$ for some $x\in [0,M)$,
and the child is connected to two (possibly different) $(h-1,M)$-trees with edges of length $x$. See~\cref{fig:htree}. A lower bound for tree distance labeling is implied by the following lemma:

\begin{figure}[h]
\begin{center}
\includegraphics[scale=0.7]{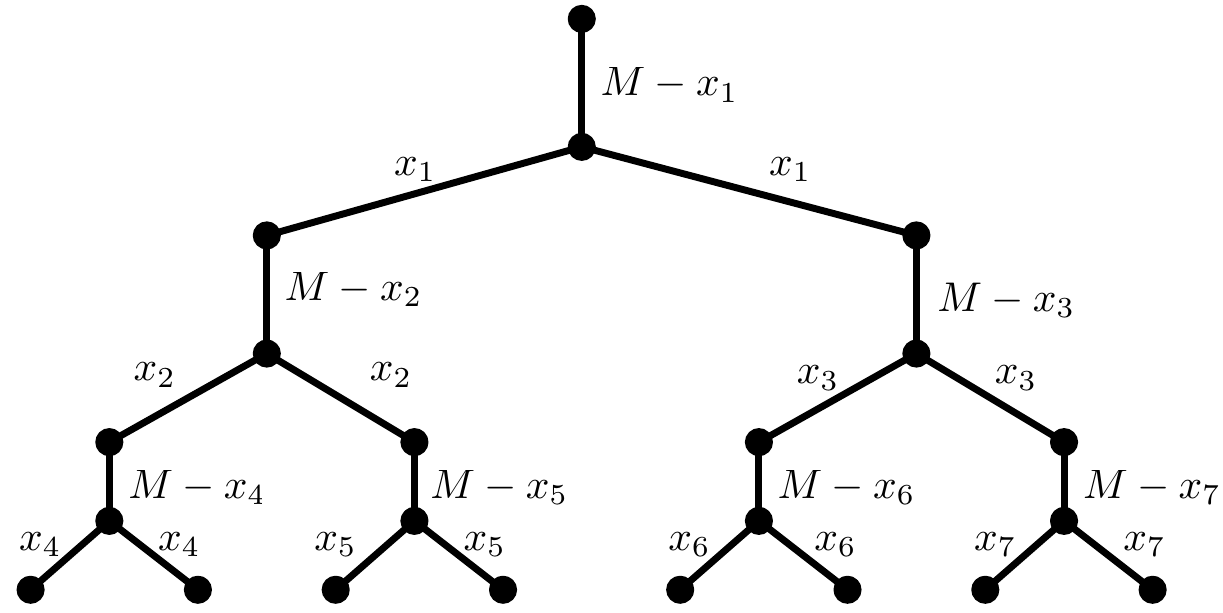}
\caption{A $(3,M)$-tree, where $x_1,\ldots,x_7\in [0,M)$.}
\label{fig:htree}
\end{center}	
\end{figure}

\begin{lemma}[{\cite{gavoille2004distance}}]
\label{lem:tree_lower_bound}
For $h\geq 1$ and $M\geq 2$, any scheme for distance labeling in $(h,M)$-trees requires labels of at least $h/2 \cdot \log M$ bits, even if we only query leaves.

\end{lemma}

\section{Distance Labeling}
\label{sec:ub}

In this section we prove~\cref{thm:ub-tree-dist}. In~\cref{sec:review} we review the labeling scheme framework of the existing solutions (in a slightly different way), and in~\cref{sec:modified} we describe our improved solution and its analysis. 

\subsection{Distance Arrays}\label{sec:review}
We now review the general framework for distance labeling. For each node  $u \in T$, consider the set of light edges $\ell_1(u), \ldots, \ell_k(u)$ along the root-to-$u$ path. 
For any light edge $e$ in the collapsed tree $\collapsed(T)$ branching from $u'$ to its child $v'$  let $\distance(e) = \distance(\hphead(u'),\hphead(v'))$. 
That is, the distance along the heavy path represented by $u'$ to the endpoint where the light edge branches and to its other end.
Let $D(u)$ denote the list $[ \distance(\ell_1(u)), \ldots, \distance(\ell_k(u)) ]$, which we call the \emph{distance array} of $u$.
The next lemma shows that designing an efficient distance labeling scheme boils down to efficiently encoding distance arrays.

\begin{lemma}
\label{lem:dist-arrays}
If we can access the elements of the distance arrays $D(u)$ and $D(v)$ then with  additional $\Oh(\log n)$ bits we can compute $\distance(u,v)$.
\end{lemma}

\begin{proof}
We first describe the additional $\Oh(\log n)$-bits. They are composed of: 
\begin{enumerate}
\item \label{list:exact_root_dist}
$\rootdist(u)$, 
\item \label{list:exact_nca}
the NCA label of $u$ generated by~\cref{lem:NCA scheme}, 
\item \label{list:exact_inorder}
the inorder number of the node $u'$ corresponding to $u$ in $\collapsed(T)$ (so that given $u,v\in T$ we can determine which node dominates the other).
\end{enumerate}

Now, suppose that $u$ is associated with $u' \in \collapsed(T)$ and $v$ with $v' \in \collapsed(T)$.
We use (\ref{list:exact_nca}) to determine $j = \lightcount(u,v) + 1$.
We assume that the inorder number of $u'$ is smaller than that of $v'$ (using (\ref{list:exact_inorder}) we can verify this and
swap $u$ with $v$ otherwise).
Thus, $u$ dominates $v$, which implies that $\rootdist(\NCA(u,v))=\sum_{i=1}^{j} \distance(\ell_i(u)) - 1$.
Recall from~\cref{sec:prelim} that $\rootdist(\NCA(u,v))$ together with $\rootdist(u)$ and $\rootdist(v)$ suffice to compute $\distance(u,v)$.
\end{proof}

\subsection{Modified Distance Arrays}
\label{sec:modified}

The main challenge remaining, is how to efficiently encode $D(u)$ for an arbitrary node $u$. 
This can clearly be done using $\Theta(\log^2 n)$ bits.  By using properties of the heavy path decomposition, Alstrup et al.~\cite{alstrup2015distance} gave a more precise bound of: $\sum_{i=1}^{k} \log \distance(\ell_i(u))  = \sum_{i=1}^{\log n} \log(n/2^i) = 1/2\log^2 n + \Oh(\log n).$
In their description, sums of the suffixes of $D(u)$ are stored instead of $D(u)$ itself, but this is essentially
the same.
Furthermore, distance arrays must be made self-delimiting by adding an additional $\Oh(\log n \log \log n)$-bits, so we get an overall space bound of $1/2\log^2 n + \Oh(\log n \log \log n)$. 

In this section, we present an improved method and analysis for encoding the distance arrays. We show that our encoding uses less space, but in the process we  lose the ability to compute the sum $\sum_{i=1}^{j} \distance(\ell_i(u))$, which is used to answer the query. 
However, in~\cref{sec:adjustment} we show that in fact a query can still be answered by adding only a small amount of auxiliary information. 
Our \emph{modified distance array} $\hat{D}(u)$ will have the following key property, which is weaker than that of the original distance array:

\begin{property}
\label{prop:mod}
Given the modified distance arrays $\hat{D}(u)$ and $\hat{D}(v)$ for leaves $u,v \in T$ such that $u$ dominates $v$, we can compute the value $\distance(\ell_j(u))$ where $j = \lightcount(u,v) + 1$.
\end{property}

\noindent

At a high level, the main idea behind the modified distance array is that, to reduce the number of bits stored for each distance $\distance(\ell_1(u)), \ldots, \distance(\ell_k(u))$ at node $u$, we potentially push some of the bits to labels of nodes dominated by $u$.
This is acceptable if our goal is to satisfy~\cref{prop:mod} since we need only  compute $\distance(\ell_i(u))$ if the other queried node $v$ is dominated by $u$.
An important observation for the analysis later is that if $\ell_i(u)$ is {\em exceptional}, we need not store $\distance(\ell_i(u))$ at all in order to satisfy~\cref{prop:mod}.  
The modified distance array consists of two parts: 
\begin{enumerate} 
\item 
a list of \emph{truncated distances} $\hat{\distance}(\ell_1(u)), \ldots, \hat{\distance}(\ell_k(u))$;
\item 
a list of \emph{accumulators} $\accum(\ell_1(u)), \ldots, \accum(\ell_k(u))$. 
\end{enumerate}
Accumulator $\accum(\ell_i(u))$ will potentially (but not necessarily) store some of the bits of the distances $\distance(\ell_i(v))$ where $v$ is a node that dominates $u$, and $i = \lightcount(u,v) + 1$.

The construction of the labels is recursive: Consider the heavy path $P$ extending from the root of $T$.  
Let $n_1,\ldots,n_{m+1}$ be the sizes of the subtrees $T_1, \ldots, T_{m+1}$ hanging from $P$ via light edges $e_1, \ldots, e_{m+1}$, where $e_{m+1}$ is the \emph{exceptional} edge.
The edges $e_1, \ldots, e_{m+1}$ are ordered according to their left-to-right order in the collapsed tree, and we use $w_1, \ldots, w_m$ to denote the nodes in $P$ from which $e_1, \ldots, e_m$ branch
($e_{m+1}$ also branches from $w_m$). See~\cref{fig:analysis}.
We use $n_1', \ldots, n_m'$ to denote the sizes of the subtrees $T_1', \ldots, T_{m}'$ rooted at nodes $w_1, \ldots, w_{m}$. For consistency, $n_{m+1}'$ denotes the size of $T_{m+1}'=T_{m+1}$.
Note that, for an arbitrary node $u \in T_i$ we have that $\ell_1(u) = e_i$.

\begin{figure}[h]
\centering
\includegraphics[scale=0.7]{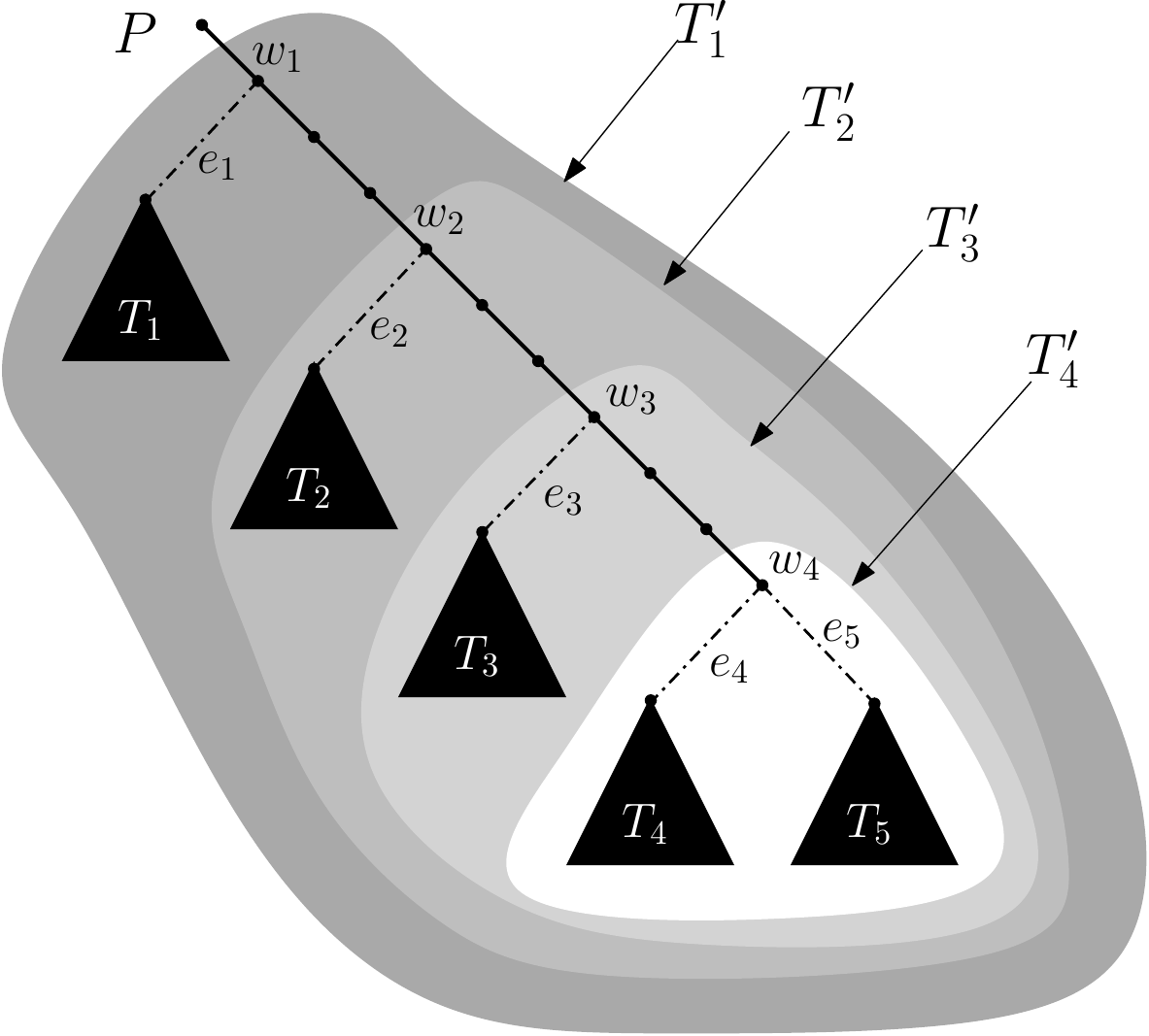}
\caption{\label{fig:analysis}A heavy path $P$ and the subtrees $T_1, \ldots, T_{m+1}$ and $T_1', \ldots, T_{m}'$. The edge $e_5$ is exceptional.}
\end{figure}

For an arbitrary node $u \in T_i$ where $i \in [1,m]$, we assume that we have some encoding of its modified distance array excluding the encoding of $\distance(e_i) = \distance(\ell_1(u))$ using $(\log^2 n_i+\log n_i \log q)/4$ bits, where
$q$ is a parameter to be fixed later.
We call this encoding the \emph{recursive problem}, and the problem of encoding $\distance(\ell_1(u))$ the \emph{top-level problem}. Recall that if $i=m+1$ (i.e., $u \in T_{m+1}$) we need not encode the distance $\distance(e_{m+1})$, since that edge is \emph{exceptional}.

We analyze the space of the top-level problem for $T'_i$ for $i =m,m-1,\ldots,1$ (i.e., from bottom to top), bounding the overall label size in terms of $n_i'$.
The goal of each iteration is to produce labels for $T_i'$ of size $(\log^2 n_i' + \log n_i' \log q)/4$. 
Consider the labels generated in the recursive problem for nodes in $T_i$ and in the previous iteration for $T_{i+1}'$ (or, if $i=m$, in the recursive problem for $T_{m+1}$).
The following two lemmas show how many bits we can spend to generate the labels for nodes in $T_i'$.
Note that these lemmas ignore the cost of making the encoding self-delimiting, as well the fact that we must take the ceiling of the bound because we cannot store a fraction of a bit. We handle these issues later.

\begin{lemma}[Slack Lemma]
\label{lem:slack}
Assume that the recursive problem for nodes in $T_i$ can be solved by storing an encoding of size $(\log^2 n_i + \log n_i \log q)/4$ bits for some parameter $q$.  
If $n_i = p\cdot n_i'$ and $p \ge 1/q$ then we can spend additional $1/2\log(1/p)\log n_i'$ bits on the top-level problem for nodes in $T_i$  to obtain an encoding of size $(\log^2 n_i' + \log n_i' \log q)/4$ bits.
\end{lemma}

\begin{proof}
To prove the lemma it is enough to calculate the difference between the size of the final encoding
and the encoding for the recursive problem:
\begin{align*}
& = (\log^2 n_i' + \log q \log n_i')/4 - (\log^2(p\cdot n_i') + \log q \log (p\cdot n_i'))/4 \\
& = (\log^2 n_i' + \log q \log n_i' -  (\log p + \log n_i')^2 - \log q \log p - \log q \log n_i')/4 \\
& = (\log^2 n_i' + \log q \log n_i' -  \log^2 p -2\log p\log n_i' -  \log^2 n_i' - \log q \log p - \log q \log n_i')/4 \\
& = (-\log^2 p -2\log p\log n_i'- \log q \log p)/4 \\
& = (2\log(1/p)\log n_i' + \log q \log(1/p) - \log^2(1/p) )/4 \\
& = (2\log(1/p)\log n_i' + \log(1/p)(\log q  - \log(1/p))/4 \\
& \geq 1/2\log(1/p)\log n_i'. &
\qedhere
\end{align*}
\end{proof}

Additionally, we have the following:

\begin{lemma}[Thin Lemma]
\label{lem:thin}
Assume that the recursive problem for nodes in $T_i$ can be solved by storing  an encoding of size $(\log^2 n_i +  \log n_i \log q) / 4$ bits for some parameter $q \ge 2$.
If $n_i = p\cdot n_i'$ and $p\le 1/2^8$ then we can spend additional $2\log n_i'$ bits on the top-level problem for nodes in $T_i$ to obtain an encoding of size $(\log^2 n_i' + \log n_i' \log q) / 4$ bits.
\end{lemma}

\begin{proof}
Similarly as in the proof of~\cref{lem:slack}, we calculate the difference:
\begin{align*}
& = (\log^2 n_i' + \log q \log n_i')/4 - (\log^2(p\cdot n_i') - \log q \log(p\cdot n_i'))/4 \\
& = (2\log(1/p)\log n_i' + \log(1/p)(\log q  - \log(1/p))/4.
\end{align*}
Now, assuming that $2\log n_i'$ is larger than the difference and using that
$p\geq 1/n_i'$ we obtain:
\begin{align*}
2\log n_i' &>  1/2\log(1/p)\log n_i' + 1/4\log(1/p)\log q - 1/4\log^2(1/p) \\
& \geq  1/2\log(1/p)\log n_i' - 1/4\log^2(1/p)  \\
& \geq 1/2\log(1/p)\log n_i' - 1/4\log(1/p)\log n_i' \\
& = 1/4\log(1/p)\log n_i'
\end{align*}
so, after dividing by $\log n_i'$, $8 > \log(1/p)$ and $p > 1/2^8$. Hence for $p \leq 1/2^8$
we can indeed use $2\log n_i'$ additional bits.
\end{proof}

We call $T_i$ \emph{thin} if $n_i\le n_i'/2^8$, and \emph{fat} otherwise.
We observe that, by the definition of the heavy path decomposition, $\log n \le \log(2n_i') \le 2\log n_i'$.
Thus, an immediate consequence of~\cref{lem:thin} is that if $T_i$ is thin, then we can afford to store $\distance(e_i)$ explicitly as $\hat{\distance}(e_i)$, without having to push any bits to the accumulators of nodes in $T_{i+1}, \ldots, T_{m+1}$.  
However, if $T_i$ is fat,~\cref{lem:slack} indicates that we do not have enough \emph{slack} to store all the bits of $\distance(e_i)$.
Instead, we store as many bits as the slack allows (rounding up to the nearest bit) in the labels of nodes $u$ in $T_i$.
We then append all the remaining bits to the accumulators $\accum(\ell_i(v))$ of all nodes $v \in \bigcup_{j = i+1}^{m+1} T_{j}$ (i.e., nodes dominated by $u$).

Because $T_i$ is fat, by the slack lemma for nodes in $T_i$ we have slack $1/2\log (n_i'/n_i) \log n_i' $
(the assumption that $T_i$ is fat allows us to adjust the constant $q$).
On the other hand, using the same calculations as in the slack lemma, the nodes in $T_{i+1}'$ have slack $1/2\log (n_{i}'/n_{i+1}')\log n_i' $: note that the size of $T_{i+1}'$ is larger than $n_i'/2^8$ by the properties of the heavy
path decomposition, as either $i<m$ and
$n'_{i+1} \geq n/2$, or $i=m$ and then $n_{i+1} \geq n_i$ so $n'_{i+1} \geq n/4-1/2$. 
Since $n_i' > n_i + n_{i+1}'$, we have that the sum $1/2(\log(n_i'/n_i) + \log(n_i'/n_{i+1}'))\log n_i'$ can be lower bounded by the minimum of $1/2(\log(1+x) + \log(1+x^{-1}))\log n_i$ for $x \in (0,\infty)$.  Thus, the slack is at least $\log n_i'$ bits in total.
However, the distance $d(\ell(u)_1)$ occupies $\log n$ bits, rather than $\log n_i'$.
As before, we can use the properties of the heavy path decomposition to bound $\log n \le \log(2n_i') = 1+\log n_i'$.
Thus, $\distance(\ell_1(u))$ occupies one extra bit more than we have accounted for with the slack.
We store this extra bit in the truncated distance $\hat{\distance}(\ell_1(u))$.
Therefore, the truncated distance $\hat{\distance}(\ell_1(u))$ consists of the most significant $\lceil 1/2\log (n_i'/n_i) \log n_i'  \rceil + 1$ bits of $\distance(\ell_i(u))$. The remaining least significant $\lfloor 1/2\log (n_i'/n_{i+1}') \log n_i'  \rfloor$ bits are concatenated to the accumulators of the nodes dominated by $u$ in $T_{i+1}, \ldots, T_{m+1}$.

For each entry in the modified distance array for a node $u$, we are pushing at most two extra bits beyond those accounted for in the slack lemma.
Thus, this works out to an additional $\Oh(\log n)$-bits in total, per label.
We make both parts of the modified distance array (the accumulators and truncated distances) self-delimiting, and also record, for each truncated distance, the number of bits pushed to the accumulators of dominated nodes.  
Overall, we end up with the following:

\begin{lemma}
\label{lem:mod-dist-array}
The modified distance array $\hat{D}(u)$ occupies at most $1/4\log^2 n + \Oh(\log n \log \log n)$ bits.
\end{lemma}

It remains to show that these modified distance arrays satisfy~\cref{prop:mod}.
To see this, consider the modified distance array for $u$ and $v$, where $j = \lightcount(u,v) + 1$, and $u$ dominates $v$.
We have stored the number of bits that were pushed to the accumulator $\accum(\ell_j(v))$ explicitly.
The starting position of this contiguous range of bits can be found by noticing that
the accumulator $\accum(\ell_j(u))$ is a suffix of $\accum(\ell_j(v))$, since all nodes that dominate $u$ also
dominate $v$.  Hence, knowing the length of the accumulator $\accum(\ell_j(u))$ allows us to
determine the starting position and, together with the explicitly stored number of pushed bits,
recover the bits themselves.
By combining them with $\hat{\distance}(\ell_j(u))$ we can reconstruct $\distance(\ell_j(u))$. 

We have therefore satisfied~\cref{prop:mod}. It remains to show why this is enough for a distance query. In~\cref{sec:adjustment} we show that it is, with only additional lower order terms to the label size. 

\subsection{Wrapping Up the Proof of Theorem~\ref{thm:ub-tree-dist}}
\label{sec:adjustment}

In order to prove~\cref{thm:ub-tree-dist} we need to show how to answer a distance query without inflating the space of~\cref{lem:mod-dist-array} by more than lower order terms.

Let $u$ be some node contained in the heavy path mapped to $u' \in \collapsed(T)$, and consider the path from $u'$ to the root of $\collapsed(T)$.
We partition this path into $B = \sqrt{\log n}$ \emph{fragments}.
The first fragment is the prefix starting at the root, denoted $f_0(u)$, and terminating at the first node $f_1(u)$ such that the subtree rooted at $\hphead(f_1(u))$ has size at most $n/2^B$. 
The $i$-th fragment is defined recursively from $f_{i-1}(u)$, ending at a node $f_i(u)$ such that the subtree rooted at $\hphead(f_i(u))$ has size at most $n/2^{iB}$, for $i \in [1, h]$, where $h=\Oh(\sqrt{\log n})$.
We explicitly store the distances $\distance(f_i(u), \treeroot(T))$ for each $i \in [1,h]$ as the \emph{fragment distance array} $F(u)$.

Next, consider a light edge $e$ in $\collapsed(T)$ that branches from the heavy path corresponding to
$u'$ to the heavy path corresponding to $v'$.
Recall that, in bounding the number of bits for the modified distance arrays, we used the fact that if the subtree rooted at $\hphead(u')$ has size $n$, then the distance, $r = \distance(\hphead(u'),\hphead(v'))$, associated with the light edge $e$ is bounded by $n$.
Instead of recording this distance $r$, for each node $u$ that stores $r$ we instead record the distance $r' = \distance( \hphead(f_j(u)), \hphead(v'))$, where $j$ is the largest index such that the subtree rooted at $\hphead(f_j(u))$ contains node $\hphead(v')$.

Obviously, $r'$ requires more bits to store than $r$, $\Oh(\sqrt{\log n})$ additional bits to be precise.
However, since there are at most $\log n$ truncated distances in $\hat{D}(u)$, we can afford to inflate each of these by $\Oh(\sqrt{\log n})$ bits. This only increases the lower order space term to $\Oh(\log^{1.5} n)$ bits.
Furthermore, for each truncated distance, we can also afford to store the corresponding index $j$ from the fragment array using $\Oh(\log n \log \log n)$ extra bits.
Thus, since~\cref{prop:mod} still holds after expanding the truncated distances, we can now recover $r'$ and read $\distance(f_j(u), \treeroot(T))$ from $F(u)$.
These two values sum to $\distance( \hphead(u'), \treeroot(T))$, which is exactly
what we wanted to compute with distance arrays.

The proof of~\cref{thm:ub-tree-dist} follows from the above $1/4 \log^2 n + o(\log^2 n)$-bit labeling scheme and the fact that we only need to label leaves and can assume $T$ is binary (see~\cref{sec:prelim}).

\subsection{Query Time Analysis}

Up until now we have not discussed how long it takes to compute the distance given two labels for nodes $u$ and $v$.  Let us summarize the steps that are required to answer a query:

\begin{enumerate}
\item Ensure $u$ dominates $v$, and swap them if that is not the case.  This can be done by examining the inorder number for $u$ and $v$, which are explicitly stored; Lemma~\ref{lem:dist-arrays} item~(\ref{list:exact_inorder}).
\item Extract the explicitly stored distances of $u$ and $v$ from the root; Lemma~\ref{lem:dist-arrays} item~(\ref{list:exact_root_dist}).
\item Compute the index $j = \lightcount(u,v) + 1$, this is done using the explicitly stored NCA encoding; Lemma~\ref{lem:dist-arrays} item~(\ref{list:exact_nca}).
\item Extract the truncated distance $\hat{\distance}(\ell_j(u))$ from array $\hat{D}(u)$.  Note that $\hat{D}(u)$ contains $\Oh(\log n)$ values, and has length $\Oh(\log^2 n)$ bits.
\item Extract the accumulator values $\accum(\ell_j(u))$ from array $\hat{D}(u)$ and $\accum(\ell_j(v))$ from array $\hat{D}(v)$.
\item Extract explicitly stored lengths of accumulator values $\accum(\ell_j(u))$ and $\accum(\ell_j(v))$.  Note that there are $\Oh(\log n)$ explicitly stored lengths, and these lengths occupy $\Oh(\log n \log \log n)$ bits.
\item Use bitwise arithmetic to extract the relevant bits of $\accum(\ell_j(v))$ which are then concatenated with $\hat{\distance}(\ell_j(u))$.  This can be done with a constant number of shifts, bitwise and/or operations, and subtractions.
\item Extract the fragment number for $j$, as well as the fragment distance from array $F(u)$.  There are $\Oh(\log n)$ fragment numbers, occupying a total of $\Oh(\log n \log \log n)$ bits, and a total of $\Oh(\sqrt{\log n})$ fragment distances, occupying a total of $\Oh(\log^{1.5} n)$ bits.
\item Compute the overall distance using addition and subtraction.
\end{enumerate}

With the exception of accessing the values stored in the various arrays just mentioned, all steps clearly take constant time.  It remains to show how to access each array element in constant time (without increasing the space bound by more than a lower order term).  First, we explicitly store the offsets of each of the (constant number of) data structures mentioned above (arrays, individual values, and the NCA labeling) for each label in a header, which is encoded using Elias $\delta$ codes in order to be self-delimiting.  This header occupies at most $\Oh(\log n)$ bits, and provides constant time access to each data structure.  Next we discuss how to access the array elements in constant time.  Earlier, we mentioned that we used Elias $\delta$ codes to delimit each array element and then concatenate their encodings.
Now, for each array that occupies $x$ bits in total and stores $y$ elements, let $p_{1} < p_{2} < \ldots < p_{y}$ be the positions of the first bit
of the encoding of each element in the concatenation. We apply Lemma~\ref{lem:encodingsequence} to this sequence. This takes
$\Oh(x\cdot \max\set{1,\log\frac{y}{x}})$ and allows us to calculate the first and the last bit of the encoding of any element in constant time.
For each of our arrays, $x = \Oh(\log n)$ and $y = \Oh(\log^{2}n)$, so storing the sequences increases the total space by only
$\Oh(\log n \log\log n)$ bits. 
Since there is a constant number of arrays, we can afford to mark the location of their corresponding sequences in the header using
$\Oh(\log n)$ bits. Then, each array access can be performed in constant time.

\subsection{Lower bound for the Level-Ancestor Problem}
\label{sec:lb}

In this section we prove~\cref{thm:lb-lvl-anc}.
The main idea of the proof is to show a lower bound for the parent problem, where the goal is to
assign a distinct label to every $u\in T$ so that given the label of $u$ we can return the label
of its parent (or a special value $\bot$ denoting that $u=\treeroot(T)$.
This is clearly a special case of the level-ancestor problem. The lower bound
is obtained by showing a correspondence between the parent problem and the following \emph{universal tree} problem: what is the size of the smallest rooted tree containing any rooted tree on $n$ nodes
as a subtree?
The connection between these two problems is captured by the following lemma.

\begin{lemma}
\label{lem:uni-tree}
If there exists a labeling scheme for the parent problem on trees of size $n$ that produces labels of size at most $S(n)$, then there exists a universal rooted tree containing all rooted trees on up to $n$ nodes
as subtrees of size $\Oh(2^{S(n)})$.
\end{lemma}
\begin{proof}
The proof is by construction.
Let $V$ be the set of all possible labels generated by the labeling scheme, and $E$ be the
directed edges between these labels defined as follows: if, a label $u$ is assigned to a node
of some tree, and $v\not=\bot$ is the label returned by the scheme for $u$, then $(u,v)$
belongs to $E$. Note that $v$ is determined solely from the bits of $u$, hence
the graph $G = (V,E)$ consists of one or more directed cycles.
See~\cref{fig:function-to-tree} (left) for an example of such a graph.
It is clear that $G$ must contain any tree $T$ on up to $n$ nodes as a subgraph, since the
labeling scheme works for all trees on $n$ nodes or less.
$G$ is not necessarily a tree itself, but we now
describe a general procedure that converts $G$ into a new graph $G' = (V',E')$ that
itself is a rooted tree, and is such that that $|V'| \le 2|V| + 1$.

Each weakly connected component of $G$ is either already a tree, or contains a cycle.
In the latter case,
we arbitrarily remove an edge $(u,v)$ from the cycle (in the figure the chosen edge
is intersected by the dashed line).  
After deleting $(u,v)$ we duplicate the entire weakly connected component, and add a new edge $(u,v')$ where $v'$ is the duplicate of $v$.
After doing this for each weakly connected component, we have increased the number of vertices to at most $2|V|$, and the resultant graph $G'$ is a forest of rooted trees.  
We add a single global root to make $G'$ a rooted tree. The total number of nodes in $G'$ is
hence at most $2|V|+1$.

Since $G$ was a universal graph for rooted trees on $n$ nodes, any rooted tree not containing the deleted edge clearly appears as a subgraph in $G'$.  
Moreover, for any rooted tree $T$ containing $(u,v)$, there exists some subpath of the cycle in $G$ which was in $T$.
Since we duplicated each node in the cycle, it is clear that any such subpath also exists in $G'$ (together with any trees rooted at nodes in the subpath), thus, $T$ appears as a subtree in $G'$.

The final detail is to consider the maximum length label output by the labeling scheme, which
consists of $S(n)$ bits.
Hence, there are at most $2^{S(n)}$ nodes in $G$ and therefore at most $\Oh(2^{S(n)})$ nodes in $G'$.
\end{proof}

\begin{figure}
\begin{center}
\includegraphics[width=0.95\textwidth]{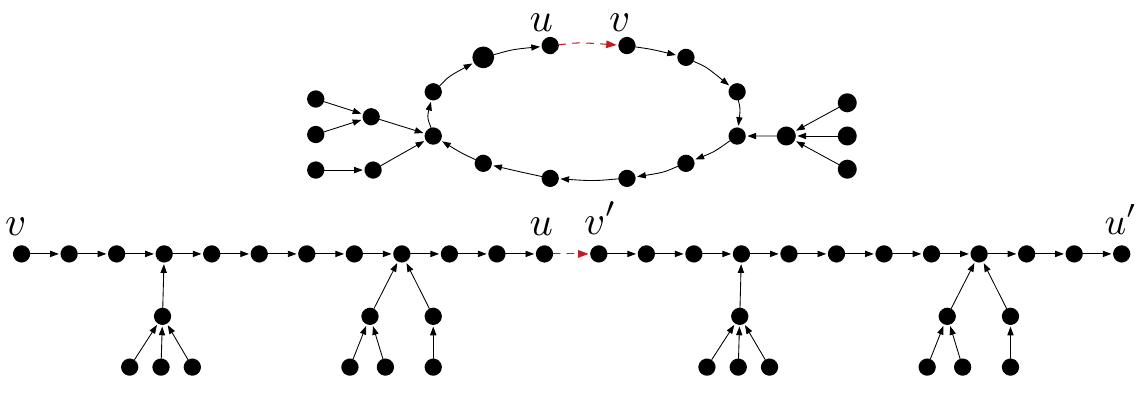}
\end{center}
\caption{\label{fig:function-to-tree}Converting a weakly connected component to a rooted tree $G'$ by duplicating the path at the dotted line.}
\end{figure}

Equipped with the previous lemma, we immediately get a lower bound on $S(n)$, provided we have a lower bound on the number of nodes in such a rooted universal tree.  Goldberg and Lifschitz~\cite{gol1968minimal} have proved very accurate bounds on the number of nodes in such rooted universal trees (see~\cite{graham1981trees} for the bound as we state it):

\begin{lemma}[\cite{graham1981trees,gol1968minimal}]
\label{lem:goldberg}
The smallest rooted tree containing all rooted trees on up to $n$ nodes as subtrees has size $n^{(\log n - 2\log \log n + \Oh(1))/2}$.
\end{lemma}

By combining Lemmas~\ref{lem:uni-tree} and~\ref{lem:goldberg},~\cref{thm:lb-lvl-anc} follows immediately.

\subsection{Effective Level-Ancestor Scheme}
\label{sec:lvl-anc}

While Alstrup et al.~\cite{alstrup2015distance} describe their scheme in terms of labeling
for distance queries, in fact it is not difficult to tweak it to obtain a scheme for level-ancestor
queries. We describe the necessary modifications to obtain a scheme for parent queries, i.e.,
assign distinct labels to every node $u\in T$ so that given the label of $u$ we can return
the label of its parent. This immediately implies a scheme for level-ancestor queries by
repeatedly moving to the parent as long as necessary.

The labeling consists of three parts. For a node $u$ on a heavy path $P$ we store:
\begin{enumerate}
\item $\distance(u, \treeroot(T))$,
\item the $\Oh(\log n)$ label generated by~\cref{lem:NCA scheme} applied on $\hphead(P)$,
\item the array $D(u)$ and, additionally, $\distance(u,\hphead(P))$.
(This is differently phrased but essentially equivalent to what the original labeling stores.)
\end{enumerate}
The labels in the NCA scheme are required to be distinct, so the labels of nodes belonging to
different heavy paths are distinct. For two nodes on the same heavy path, storing
$\distance(u,\hphead(P))$ explicitly ensures that their labels are not the same.
Each label consists of $1/2\log^2 n+\Oh(\log n)$ bits, because of the bound on the encoding
of $D(u)$.
We need to argue that given the label of $u \not=\bot$ we can construct the
label of its parent. 

The NCA labeling scheme from~\cref{lem:NCA scheme} has the property that the label of every
node $u$ is a concatenation of heavy and light labels $h_0.\ell_1.h_1.\ell_2 \ldots \ell_k.h_k$.
These labels uniquely determine the path from the root to $u$: $h_0$ encodes how far along the
heavy path starting at the root we should continue. Then, either $k=0$ and $u$ in fact
lies on the heavy path starting at the root, or $\ell_1$ encodes which light edge outgoing
from the current node should be followed. Finally, $h_1.\ell_2\ldots \ell_k.h_k$
recursively encodes the remaining part of the path to $u$ in the subtree hanging off the heavy path
starting at the root. It is not necessarily true that given the NCA label of a node $u$
we can determine the NCA label of its parent. However, by truncating the NCA label of $u$
we can construct the NCA label of the parent of $\hphead(P)$.

Given the label of $u$, we construct the label of its parent $u'$ as follows.
$\distance(u, \treeroot(T))$ needs to be decreased by 1. Then we inspect $\distance(u,\hphead(P))$.
If $u\neq \hphead(P)$, we decrease $\distance(u,\hphead(P))$ by 1 and are done.
Otherwise, we can use the NCA label of $u$ to determine the label of its parent as
explained above. Let $P'$ be the heavy path of $u'$. The last element of the array
$D$ is $\distance(u,\hphead(P'))$, so by subtracting 1 we obtain
$\distance(u',\hphead(P'))$. Finally, we remove the last element of $D$.

\section{\boldmath$k$-Distance Labeling}
\label{sec:k}

In this section we prove~\cref{thm:krestricted}. Recall that in $k$-distance labeling,  given the labels of $u$ and $v$ we need to decide if the length of the $u$-to-$v$ path is at most $k$, and if so return it. 

\subsection{Lower Bound for Small \boldmath$k$}
\label{sec:smallk}

We define a family of trees and show that labeling the leaves of all trees in that family for
$k$-distance queries requires $\log n + \Omega (k \cdot \log \frac{\log n}{k\log k})$-bits.

An $\vec{x}$-regular tree, where $\vec{x} = (x_1,\cdots,x_k) \in \mathbb{N}^k$, is a rooted tree
of height $k$ where all depth-$i$ nodes have the same degree $x_{i+1}$.
An $(\vec{x}, h, d)$-regular tree, where $(x_1,\cdots,x_k) = \vec{x} \in [h]^k$, is a 
$\vec{y}$-regular tree with $\vec{y} = (d^{x_1}, d^{h - x_1},d^{x_2}, d^{h - x_2},\cdots, d^{x_k}, d^{h-x_k})$.
The total number of leaves in such a tree is $d^{k\cdot h}$. See~\cref{fig:reg2} for an example.

\begin{figure}[h]
\centering	
\begin{tikzpicture}[scale=0.5,nodes={draw,circle,fill,inner sep=0pt,minimum size=5pt},thick]
	\def \x{5}
	\def \y{2.5}
	\def \z{0.5}
    \node{}
	[sibling distance=\x cm]
    child{node {}
		[sibling distance=\y cm]
		child{node{}
			[sibling distance=\z cm] 
			child{node{} child{node{}}}
			child{node{} child{node{}}}
			child{node{} child{node{}}}
			child{node{} child{node{}}}
		}
		child{node{}
			[sibling distance=\z cm]
			child{node{} child{node{}}}
			child{node{} child{node{}}}
			child{node{} child{node{}}}
			child{node{} child{node{}}}
		}
	}	
    child{node {}
		[sibling distance=\y cm]
		child{node{}
			[sibling distance=\z cm]
			child{node{} child{node{}}}
			child{node{} child{node{}}}
			child{node{} child{node{}}}
			child{node{} child{node{}}}
		}
		child{node{}
			[sibling distance=\z cm]
			child{node{} child{node{}}}
			child{node{} child{node{}}}
			child{node{} child{node{}}}
			child{node{} child{node{}}}
		}
	}	
;
\end{tikzpicture}
\caption{$(\vec{x},d,h)$-regular tree with $\vec{x} = (1,2)$ and $d=h=2$}	
\label{fig:reg2}
\end{figure}
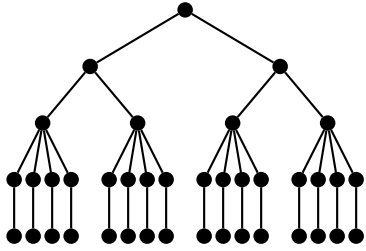

We consider $(\vec{x},h,d)$-regular
trees for some parameters $h$ and $d$ to be chosen later. Consider a labeling scheme
that assigns a label to every leaf of such a tree for $2k$-distance queries.
The following lemma shows that a $(\vec{x},h,d)$-regular tree and a $(\vec{y},h,d)$-regular tree cannot share many identical labels. More formally, let
$\common(\vec{x},\vec{y})$ denote the maximum number of labels that can be used in both trees. The following is  an upper bound on the sum of $\common(\vec{x},\vec{y})$.

\begin{lemma}\label{lem:common}
$\sum_{\vec{x}, \vec{y} \in [h]^k}{\common (\vec{x},\vec{y})} \le  \left(h\cdot d^h \left(1 + \frac{2}{d-1}\right)\right)^k$. 
\end{lemma}

\begin{proof}\belowdisplayskip=-12pt
We first prove that $\common(\vec{x},\vec{y}) \leq \prod_{i=1}^{k}{d^{\min\set{x_i, y_i}}d^{h - \max\set{x_i,y_i}}}$.

By asking all $2k$-distance queries between a specified subset $S$ of leaves of
the $(\vec{x},h,d)$-regular tree we can recover the shape of the subtree induced by $S$.
Hence, if two trees share $\common(\vec{x},\vec{y})$ labels, then they must have
a common isomorphic subtree on $\common(\vec{x},\vec{y})$ leaves. To bound the
maximum number of leaves in such a subtree, observe that the degree of a node
at depth $2i-2$ is at most $\min\set{d^{x_i},d^{y_i}}$, and the degree of
a node at depth $2i-1$ is at most $\min\set{d^{h-x_i},d^{h-y_i}}$. The maximum number
of shared labels is hence the product of all these quantities over $i=1,\ldots,k$. 
We conclude that that $\common(\vec{x},\vec{y}) \leq \prod_{i=1}^{k}{d^{\min\set{x_i, y_i}}d^{h - \max\set{x_i,y_i}}}$. It then follows that
\begin{align*}
\sum_{\vec{x}, \vec{y} \in [h]^k}{\common (\vec{x},\vec{y})} & = 
\sum_{\vec{x}, \vec{y} \in [h]^k}\prod_{i=1}^{k}{d^{\min\set{x_i, y_i}}d^{h - \max\set{x_i,y_i}}} \\
& = \prod_{i=1}^{k}{\sum_{1 \leq x,y \leq h}{d^{\min\set{x, y}}d^{h - \max\set{x,y}}}} 
	\\
& = \prod_{i=1}^{k}{(h\cdot d^h + 2\sum_{x < y}{d^x d^{h - y}})} \\
& = \prod_{i=1}^{k}{(h\cdot d^h + 2\sum_{x = 1}^{h-1}{d^x\sum_{y=0}^{h-x-1}{d^{y}}})}
	\\
& = \prod_{i=1}^{k}{(h\cdot d^h + 2\sum_{x = 1}^{h-1}{d^x\frac{d^{h-x}-1}{d-1}})}
	\\
& \leq \prod_{i=1}^{k}({h\cdot d^h + 2\cdot h \frac{d^h}{d-1})}
	\\
& = \left(h\cdot d^h \left(1 + \frac{2}{d-1}\right)\right)^k.
\end{align*}
\end{proof}

\noindent Since the total number of leaves in the $(\vec{x},h,d)$-regular trees  family is $d^{k\cdot h} \cdot h^k$, the number of distinct labels required to label them is thus at least:

$$
d^{k\cdot h} \cdot h^k -\sum_{\vec{x} < \vec{y}}{\common(\vec{x},\vec{y})}  =
d^{k\cdot h} \cdot h^k - \frac{1}{2}\sum_{\vec{x} \neq \vec{y}}{\common (\vec{x},\vec{y})} 
 = \frac{3}{2} \cdot d^{k\cdot h} \cdot h^k - \frac{1}{2}\sum_{\vec{x}, \vec{y} \in [h]^k}{\common (\vec{x},\vec{y})}.
$$

\noindent Now we set $d=2k+1$, and since $(1 + \frac{1}{k}) \leq e^{\frac{1}{k}}$ we have from~\cref{lem:common} that
$\sum_{\vec{x}, \vec{y}}{\common (\vec{x},\vec{y})}  \leq e \cdot d^{k\cdot h} \cdot h^k$,
so the number of unique labels is at least:
$(3/2-e/2) \cdot d^{k\cdot h} \cdot h^k > 0.1 \cdot d^{k\cdot h} \cdot  h^k$.
Setting $n=d^{k\cdot h}$ this becomes $0.1 \cdot  h^k \cdot n$,
making the number of required bits at least:
\[
\log n + k \log h - \Oh(1) = \log n + k \log \frac{\log n}{k\log d}-\Oh(1) = \log n + \Omega(k \cdot \log \frac{\log n}{k \log k}).
\]
Note that for the above calculation to make sense, we need that $d^{k} \leq n$.

\subsection{Lower Bound for Large \boldmath$k$}
\label{sec:largek}

The lower bound from~\cref{sec:smallk} is not meaningful for large values of $k\in [\log n,n]$.
In this section we show that the lower bound of Gavoille et al.~\cite{gavoille2004distance} for general distance queries, can be translated into a lower bound
of $\Omega(\log n\cdot \log(k/\log n) )$ for $k$-distance queries.

The lower bound uses the family of $\htree$s (see~\cref{sec:prelim}). Recall that every edge of an $\htree$ has a weight from
$[0,M]$. It is easy to verify that the number of nodes in such a tree is $3\cdot 2^h - 2$,
hence the distance between any two leaves is no more than $2hM$.

If $M\leq k/(2h)$ then, because the distance between any two leaves
in the tree is at most $2hM\leq k$, any labeling of the leaves for $k$-distance can be used for general distance labeling. By~\cref{lem:tree_lower_bound}, such a labeling scheme would require labels of at least $h/2 \cdot \log M$-bits.
We set $h=\log \sqrt{n/3}$ and $M = \min\set{k/2h,2^h}$. Then, by subdividing the edges
of an $\htree$ we obtain an unweighted tree on at most $n$ nodes. Labeling the leaves
of such a tree for $k$-distance can be used for general distance labeling of the $\htree$,
so we obtain the following lower bounds:
\begin{enumerate}[label=(\arabic*)]
\item
if $\frac{k}{2h} \leq 2^h$, the number of required bits is $\frac{h}{2}\cdot \log\frac{k}{2h} = \Omega(\log n\cdot \log\frac{k}{\log n})$;
\item
if $\frac{k}{2h} > 2^h$, the number of required bits is $\frac{h}{2}\cdot h = \Omega(\log^2 n)$,
so $\Omega(\log n\cdot \log\frac{k}{\log n})$ for $k\leq n$.
\end{enumerate}

\subsection{Upper Bound}

In this section we present our improved upper bound for $k$-distance labeling. We build upon the
ideas of Alstrup, Bille, and Rauhe~\cite{alstrup2005labeling}, who presented an
$\log n+\Oh(k^2\log(k\log n))$ bits labeling scheme. As a preliminary step, we will show
an $\Oh(\log k\cdot \log n)$ bits scheme for $k \geq \log n$, and then move to the more complicated
$\log n + \Oh(k\log\frac{\log n}{k})$ bits scheme for $k < \log n$. 

Consider the heavy path decomposition of $T$. We define the \emph{light range} of $u$, denoted $\lightrange{u}$, to contain the preorder number of all nodes in $T_u$ if $u$ has no heavy child, and all nodes
in $T_u\setminus T_{\heavy(u)}$ otherwise. 
We say that $v$ is a {\em significant ancestor} of $u$ if $\pre(u) \in \lightrange{v}$. For example, in~\cref{fig:heavy-path-decomp} $v$ is a significant ancestor of $u$ since the light range of $v$ is $\lightrange{v} = [5,23)$.
 The number of significant ancestors of $u$ is equal to $\lightdepth(u)=\Oh(\log n)$.
The nearest common significant ancestor of $u$ and $v$, denoted $\NCSA(u,v)$, is
$w$ such that $\pre(w)$ is as large as possible and $w$ is a significant ancestor of both $u$ and $v$.
In other words, $w$ is the first significant ancestor on the path from $u$ to the root, which is also
a significant ancestor of $v$. The heavy path $P$ such that $\hphead(P)$ is a child of
$\NCSA(u,v)$ is called the nearest common heavy path of $u$ and $v$ and denoted
$\NCH(u,v)$. When there is no common significant ancestor for $u$ and $v$ we set $\NCSA(u,v)$ to $\nil$ and $\NCH(u,v)$ to be the heavy path starting at the root. 

Let the significant ancestors of $u$ and $v$ on $\NCH(u,v)$ be $u'$ and $v'$,
respectively. Then $\distance(u,v)=\distance(u,u')+\distance(u',v')+\distance(v,v')$.
Computing $\distance(u,v)$ consists of two steps:
\begin{enumerate}
\item identifying $\NCH(u,v)$, $u'$ and $v'$, and computing $\distance(u,u')$ and $\distance(v,v')$,
\item computing $\distance(u',v')$.
\end{enumerate}
We describe these steps separately, and then describe how to implement them in constant time.

\paragraph{Identifying \boldmath$\NCH(u,v)$.} 

For an integer range $A = [a,b]\subset [1,n]$ we define its identifier $\id(A)$ by considering a binary trie representing all
words of length $\ceil{\log n}$. The label of a node $u$ in the trie is the concatenation of the labels of the edges
on the path from the root to $u$. Every integer $x\in [1,n]$ corresponds to a leaf $u$ in the trie, such that
the label of $u$ is the binary expansion of $x$. Then, $\NCA(a,b)$ is the nearest
common ancestor of the leaves corresponding to $a$ and $b$ in the trie, $\height(A)$ is the height of the subtree rooted
at $\NCA(a,b)$, and finally $\id(A)$ is the label of $\NCA(a,b)$.

\begin{observation}
\label{ob:range_id}
For any range $A$:
\begin{enumerate}
	\item \label{ob:range_id:compute}
	$\id(A)$ can be computed given $\height(A)$ and any $x\in A$, 	
	\item \label{ob:range_id:disjoint}
	$A \cap B = \emptyset \implies \id(B) \neq \id(A)$.
\end{enumerate}
\end{observation}

 Alstrup, Bille, and Rauhe~\cite{alstrup2005labeling} use the notion of
\emph{significant preorder numbers}. We replace it with our notion of range identifier,
that has very similar properties, yet is somewhat easier to operate on (and hence we are able to achieve
much better query time). For any node $u\in T$, let $\id(u)=(\id(\lightrange{u}), \lightdepth(u))$.	

\begin{lemma}
\label{lem:uniqueid}
For any nodes $u,v\in T$, if $u\neq v$ then $\id(u)\neq \id(v)$.
\end{lemma}

\begin{proof}	
If $\lightdepth(u) \neq \lightdepth(v)$ then we are done.
Otherwise, $\lightrange{v}$ and $\lightrange{u}$ are disjoint, so by \cref{ob:range_id}.\ref{ob:range_id:disjoint} $\id(\lightrange{v}) \neq \id(\lightrange{u})$ and we are also done.
\end{proof}

Consider a node $u\in T$ and let $u=u_0,u_1,u_2,\ldots$ be all of its significant ancestors in the order
in which they appear on the path from $u$ to $\treeroot(T)$. Let $u_{r}$ be the last of these
ancestors such that $\distance(u,u_{r}) \leq k$.
We call $u_{r}$ the top significant ancestor of $u$. The label of $u$ consists of $\pre(u)$, $\lightdepth(u)$,
and an encoding of $\height(\lightrange{u_i})$ for every $i=0,1,\ldots,r$. By~\cref{ob:range_id}.\ref{ob:range_id:compute}
this is enough to compute $\id(u_i)$ for every $i=0,1,\ldots,r$. Consequently, given the
labels of $u$ and $v$, we can either detect that the distance from $u$ or $v$ to
$\NCA(u,v)$ exceeds $k$, or calculate $\lightdepth(\NCSA(u,v))$.

To encode $\height(\lightrange{v_i})$ for every $i=0,1,\ldots,r$, we observe that
$\lightrange{v_i}\subseteq \lightrange{v_{i+1}}$ and that $r\leq \min\set{\log n,k}$. Hence, we need to encode
a non-decreasing sequence of $\min\set{\log n,k}$ numbers from $[0,\log n]$. 
By~\cref{lem:encodingsequence}, for $k < \log n$ this can be done using
$\Oh(k\log\frac{\log n}{k})$ bits and for $k \geq \log n$
using $\Oh(\log n)$ bits, and allows us to calculate $\lightdepth(\NCSA(u,v))$ or detect
that $\distance(u,v)>k$.

We encode in the label of $u$ the distance from $u$ to $u_i$ for every
$i=0,1,\ldots,r-1$. Because $0=d(u,u_0)<d(u,u_1)<\cdots < d(u,u_{r-1})\leq k$
we need to encode an increasing sequence of $\min\set{\log n,k}$ numbers from the range
$[0,k]$. By~\cref{lem:encodingsequence},
if $k<\log n$ this can be done using $\Oh(k)$ bits and if $k\geq \log n$
using $\Oh(\log n\cdot\log \frac{k}{\log n})$ bits. Then, after having found $\lightdepth(\NCSA(u,v))$
we can compute $\distance(u,u')$ and $\distance(v,v')$.
 
\begin{figure}[h]
\begin{center}
\includegraphics[scale=0.6]{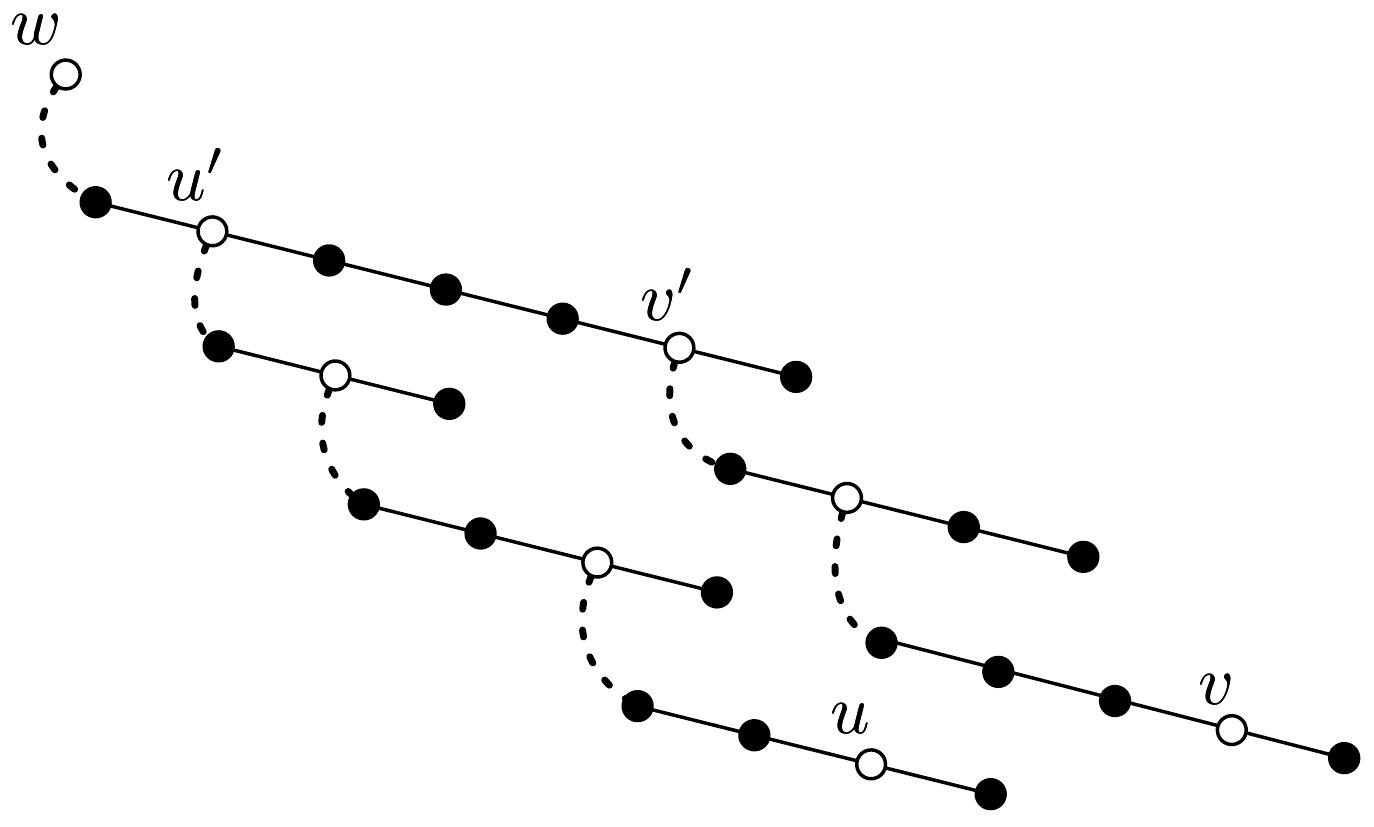}
\caption{$w=\NCSA(u,v)$, $u'$ and $v'$ are the significant ancestors of $u$ and $v$ on $\NCH(u,v)$,
respectively. Significant ancestors are white,  heavy
edges are solid, and light edges are dashed.}
\label{fig:ncsa}
\end{center}	
\end{figure}

\paragraph{Computing \boldmath$\distance(u',v')$.}

Recall that $u'$ and $v'$ are the significants ancestors on the $\NCH(u,v)$ of $u$ and $v$,
respectively. We want to compute $\distance(u',v')$. If $u'$ is not the top significant ancestor of
$u$ and $v'$ is not the top significant ancestor of $v'$ then from the distances encoded in the
labels of $u$ and $v$ we can retrieve $\distance(u',\NCSA(u,v))$ and $\distance(v',\NCSA(u,v))$,
and  return their absolute difference as $\distance(u',v')$.
Now consider the case that $u'$ is the top significant ancestor of $u$, but $v'$ is not the top significant
ancestor of $v$. To deal with this case, the label of $u$ should also encode
the distance $\alpha$ from
$u'$ to the head of it's heavy path. This distance might be very large (even up to $n$),
so we cap it at $2k+1$ to use only $\Oh(\log k)$ bits. Since 
$v'$ is not the top significant ancestor of $v$, we can retrieve $\beta=\distance(v',\NCSA(u,v))$ as in the previous case.
We know that $\beta \leq k$ because otherwise $v'$ would be the top significant ancestor
of $v$. Recall that $\distance(u',v')$ is equal to the absolute difference between 
$\distance(u',\NCSA(u,v))$ and $\distance(v',\NCSA(u,v))$.
If $\alpha=2k+1$ then this value must
exceed $k$, so we terminate. Otherwise, we return $|\alpha-\beta|$.

The remaining and most complicated case is when $u'$ is the top significant ancestor of $u$
and $v'$ is the top significant ancestor of $v'$. If $k > \log n$, the solution is simple,
as we can afford to store the distance from the top significant ancestor to the head of its
heavy path for every node (i.e., $\distance(u,u_{r+1})$) using $\Oh(\log n)$ bits.
The rest of this section is dedicated for solving $k\leq \log n$.

To make the further exposition more concise, we define the 2-approximation of an integer $x$, denoted $\appx{x}$,
as the largest power of 2 not exceeding $x$. That is, $\appx{x} = 2^{\floor{\log x}}$. Clearly, 
2-approximation is monotone, meaning that $x\leq y$ implies $\appx{x}\leq \appx{y}$,
and furthermore $\appx{x} < \appx{2x}$.

\begin{lemma}
\label{lem:ranges}
Let $A,B,C$ be three open intervals such that $A\cap B=\emptyset$ and $A,B\subseteq C$.
Then $\appx{|C|}\neq \appx{|A|}$ or $\appx{|C|}\neq \appx{|B|}$.
\end{lemma}

\begin{proof}
Assume that $|B| \leq |A|$. Then $2|B| \leq |A|+|B| \leq |C|$ and by the properties of
2-approximation $\appx{|B|} < \appx{|C|}$, so indeed $\appx{|B|}\neq \appx{|C|}$.
Symmetrically, if $|A| \leq |B|$ then  $\appx{|A|}\neq \appx{|C|}$.
\end{proof}

The following lemma captures the essence of the $k$-distance scheme of Alstrup, Bille, and
Rauhe~\cite{alstrup2005labeling}, while being optimized so that we can obtain our improvement.

\begin{lemma} \label{lem:monotone_sequence_scheme}
Consider an increasing sequence of integers $a_1 < a_2 < \ldots < a_s$. Given $a_i < a_j$,
$i' = i\bmod{k}$, $j' = j\mod{k}$, and $\appx{a_{i+t} - a_i}$ and $\appx{a_j - a_{j-t}}$ for every
$t=1,2,\ldots,k$ we can calculate $j-i$ or determine that $j-i > k$ in constant time.
\end{lemma}

\begin{proof}
We start by setting $t = j'-i'$.
Now either $t = j - i$ or $j - i \geq k+t$. Hence we only need to distinguish
between these two cases.

Consider three intervals $(a_i, a_{i+t})$, $(a_{j-t}, a_j)$ and $(a_i, a_j)$. If $t=j-i$
then these three intervals are equal and so are $\appx{a_{i+t}-a_i}, \appx{a_j-a_{j-t}}$
and $\appx{a_j-a_i}$.
Otherwise $j-i \geq k+t >  2t$, so $(a_i, a_{i+t})$ and $(a_{j-t},a_j)$ 
are two disjoint intervals contained
in $(a_i, a_j)$. Hence by~\cref{lem:ranges} either $\appx{a_{i+t}-a_i}\neq \appx{a_j-a_i}$
or $\appx{a_j-a_{j-t}}\neq \appx{a_j-a_i}$.
Therefore after retrieving $\appx{a_{i+t}-a_i}$ and $\appx{a_j-a_{j-t}}$ and calculating
$\appx{a_j-a_i}$ we can distinguish between the two cases and either return $j-i$
or report that $j-i > k$. Notice that $\appx{a_j-a_i}$ can be calculated in constant
time using standard word-RAM operations.	
\end{proof}

We need to show that, for every heavy path, we store enough  information  for applying~\cref{lem:monotone_sequence_scheme}.
Consider a heavy path $u_1 - u_2 - \ldots - u_s$, where $u_1$ is the head.
By the properties of the heavy path decomposition,
$\id(\lightrange{u_1})<\id(\lightrange{u_2})<\ldots<\id(\lightrange{u_s})$.
The label of every $u\in T$ such that $u_i$ is the top significant ancestor of $u$ encodes the following:
\begin{enumerate}
\item 
$\id(\lightrange{u_i})$;
\item 
$\appx{\id(\lightrange{u_{i+t}}) - \id(\lightrange{u_i})}$ and $\appx{\id(\lightrange{u_i}) - \id(\lightrange{u_{i-t}})}$ for every $t=1,2,\ldots,k$;
\item 
$i\bmod{k}$.
\end{enumerate}
To encode $\id(\lightrange{u_i})$, we store $\height(\lightrange{u_i})$ using $\Oh(\log\log n)$ bits.
Encoding $\appx{\id(\lightrange{u_{i+t}}) - \id(\lightrange{u_i})}$ and
$\appx{\id(\lightrange{u_i}) - \id(\lightrange{u_{i-t}})}$ for every $t=1,2,\ldots,k$ reduces
to encoding two non-decreasing sequences of $k$ integers from $[0,\log n]$.
By~\cref{lem:encodingsequence}, such a sequence can be stored
using $\Oh(k\log\frac{\log n}{k})$-bits.
Finally, $i\bmod{k}$ is encoded using $\Oh(\log k)$ bits.
Notice that both $\Oh(\log\log n)$ and $\Oh(\log k)$ are absorbed by $\Oh(k\log\frac{\log n}{k})$.

To conclude, given the labels of $u$ and $v$, whose significant
ancestors $u'$ and $v'$ are on $\NCH(u,v)=u_1-u_2-\ldots -u_s$ and are both the top significant ancestors,
we can now calculate $\distance(u',v')$ or detect that it exceeds $k$
by retrieving the necessary information from the labels of $u$ and $v$ and then
applying~\cref{lem:monotone_sequence_scheme}.  
Finally, in the following section (\ref{sec:query-time}) we show that queries can be supported in constant time.
The gist of the improvement in the query time is that $\id(\lightrange{u_i})$ can be obtained
from $\pre(u)$ by truncating the last $\height(\lightrange{u_i})$ trailing bits and setting
the $\height(\lightrange{u_i})^\text{th}$ bit to $1$.

\subsection{Query Time Analysis}
\label{sec:query-time}

We now show how to implement the query in constant time. The main difficulty is in
determining $\lightdepth(\NCSA(u,v))$ efficiently. Once it is known, from $\lightdepth(u)$
and the encoding of the distances from $u$ to its significant ancestors implemented
with~\cref{lem:encodingsequence} we obtain
$\distance(u,u')$ in constant time, and similarly for
$\distance(v,v')$ (or conclude that $\distance(u,v)$ exceeds $k$). Calculating $\distance(u',v')$
requires invoking~\cref{lem:monotone_sequence_scheme} while providing
access to the stored non-decreasing sequences of 2-approximations
with~\cref{lem:encodingsequence},
so also takes only constant time.

Recall that the label of $u$ contains $\pre(u)$, $\lightdepth(u)$, and an encoding of the sequence
$\height(\lightrange{u_0}) \leq \height(\lightrange{u_1}) \leq \dots \leq \height(\lightrange{u_r})$ implemented with~\cref{lem:encodingsequence}.
Similarly, the label of $v$ contains $\pre(v)$,
$\lightdepth(v)$, and $\height(\lightrange{v_0}) \leq \height(\lightrange{v_1}) \leq \dots \leq \height(\lightrange{v_s})$.
We want to calculate $\lightdepth(\NCSA(u,v))$. For now, we assume that $r=s$ and
$\lightdepth(u_i)=\lightdepth(v_i)$ for every $i=0,1,\ldots,r$.
Then, calculating $\lightdepth(\NCSA(u,v))$ reduces to finding the smallest $i$
such that $u_i=v_i$. Notice that then $u_j=v_j$ for every $j=i,i+1,\ldots,r$.
If $u_j=v_j$ then clearly $\height(\lightrange{u_j})=\height(\lightrange{v_j})$, so we start
with locating the smallest $i'$ such that $\height(\lightrange{u_j})=\height(\lightrange{v_j})$
for every $j=i',i'+1,\ldots,r$. This can be done in constant time by computing the
longest common suffix of both sequences.

Because $\lightdepth(u_j)=\lightdepth(v_j)$ for every $j=0,1,\ldots,r$,
it remains to find the smallest $i \geq i'$ such that
$\id(\lightrange{u_j})=\id(\lightrange{v_j})$ for every $j=i,i+1,\ldots,r$.
Observe that $\id(\lightrange{u_j})$ is obtained by clearing all
$\height(\lightrange{u_j})$ least significant bits of $\pre(u)$ and, if $\height(\lightrange{u_j})>0$, setting the
$\height(\lightrange{u_j})^\text{th}$ bit to 1, and similarly for $\id(\lightrange{v_j})$.
Without loss of generality, assume that $\height(\lightrange{u_{i'}})=\height(\lightrange{v_{i'}})>0$
(if not, $i=i'$ is checked separately in constant time). We find the longest
common prefix of the binary expansions of $\pre(u)$ and $\pre(v)$, i.e., the smallest
$\ell \geq 0$ such that their binary expansions
are the same after truncating the $\ell$ least significant bits.
$\ell$ can be found in constant time using standard word-RAM operations
$\text{MSB}(\pre(u) \text{ XOR } \pre(v))$.
Then, for $\id(\lightrange{u_i})=\id(\lightrange{v_i})$ to hold, we need to clear
at least $\ell$ least significant bits of $\pre(u)$ and $\pre(v)$.
Hence it remains to find the smallest $i \geq i'$ such that
$\height(\lightrange{u_i})=\height(\lightrange{v_i}) \geq i$. Such an $i$ can be found in constant time
with a successor query on the encoded sequence.

If $r\neq s$ or $\lightdepth(u_0)\neq \lightdepth(v_0)$, then essentially
the same argument works, except that we need to compute the longest common
prefix of suffixes instead of whole sequences.

\section{Approximate Distance Labeling}
\label{sec:approx}

In this section we prove~\cref{thm:approximate}. Recall that in  $(1+\eps)$-approximate distance labeling,  given the labels of $u$ and $v$ we need to output some value in the interval $[\distance(u,v),(1+\eps)\cdot\distance(u,v)]$. 

\subsection{Lower bound}

To show the lower bound we modify the family of $(h,M)$-trees such that exact distances between leaves can be inferred from their approximate distances. Thereafter, we can invoke~\cref{lem:tree_lower_bound} to establish the lower bound.

An $(h,M)$-tree is modified by first subdividing its edges to obtain an unweighted tree of height $h\cdot M$. The edges of this unweighted tree are then further subdivided: every edge of depth $d\ge 0$ is subdivided into $\floor{(1 + \eps)^{hM-d}}$ edges. Note that in the original $(h,M$)-tree all leaves are at the same distance from the root. Therefore, if the distance between two leaves is $2k$ in the original tree, it is $f(k)=2\sum_{i=1}^{k}{\floor{(1+\eps)^i}}$ in the final tree.
A $(1+\eps)$-approximation of this distance belongs to the interval $[f(k), (1+\eps)f(k)]$. 
We next show that these intervals are disjoint, so in fact a $(1+\eps)$-approximation 
of $f(k)$ is enough to infer the original distance, $2k$.

Observe that $f(k)$ is monotone, so to prove that the intervals $[f(k), (1+\eps)f(k)]$ are
disjoint, it is enough to show that $(1+\eps)f(k) < f(k+1)$, or:
\begin{align*}
(1+\eps)\sum_{i=1}^{k}{\floor{(1+\eps)^i}} &< \sum_{i=1}^{k+1}{\floor{(1+\eps)^i}}, \text{ or equivalently} \\
\eps\sum_{i=1}^{k}{\floor{(1+\eps)^i}} &< \floor{(1+\eps)^{k+1}}. \\
\end{align*}
Since $\floor{(1+\eps)^i} < (1+\eps)^i$, it is enough to show that:
\begin{align*}
\eps\sum_{i=1}^{k}{(1+\eps)^i} &< \floor{(1+\eps)^{k+1}}, \text{ or equivalently}\\	
(1+\eps)^{k+1} - (1+\eps) &< \floor{(1+\eps)^{k+1}}.\\	
\end{align*}
Since $x - 1 < \floor{x}$ is always true, we conclude that the intervals are indeed disjoint.
Hence, by~\cref{lem:tree_lower_bound} we obtain that labeling the leaves of the final tree
for $(1+\eps)$-approximate distances requires $h/2 \cdot \log M$ bits. It remains to choose $h$ and $M$ and 
rephrase this bound in terms of the size of the final tree. The size of the final tree is at most
\begin{align*}
2\sum_{i=0}^{h-1}{2^{h-1-i}} \sum_{j=M\cdot i+1}^{M\cdot(i+1)}\floor{(1 + \eps)^j} & =
\sum_{i=0}^{h-1}2^{h-i} \sum_{j=M\cdot i+1}^{M\cdot(i+1)} (1 + \eps)^j \\
&\leq 2^{h}\sum_{i=0}^{h-1}{2^{-i}}(1+\eps)^{M\cdot i+1}\frac{(1+\eps)^{M}-1}{(1+\eps)-1}\\
&\leq 2^h \frac{1}{\eps}\sum_{i=0}^{h-1}{2^{-i}}(1+\eps)^{M(i+1)+1} \\
&\leq 2^h \frac{1}{\eps}(1+\eps)^{M+1}\sum_{i=0}^{h-1}\left(\frac{(1+\eps)^{M}}{2}\right)^{i} \\
&= 2^h \frac{1}{\eps}(1+\eps)^{M+1} \frac{(\frac{(1+\eps)^M}{2})^h-1}{\frac{(1+\eps)^M}{2}-1}\\
&\leq \frac{2}{\eps}\frac{(1+\eps)^{M+1}}{(1+\eps)^M-2} (1+\eps)^{M\cdot h} \\
\end{align*}
We set $M=2/\eps$. Then, because $\eps \leq 1$ and $(1+\eps)^M \geq 4$, the size
is at most:
\begin{align*}
&\leq 2\frac{1+\eps}{\eps} \frac{(1+\eps)^M}{(1+\eps)^M-2} e^{2h} \\
&\leq \frac{8}{\eps} e^{2h}.
\end{align*}
We set $h = \log (\eps \cdot n/8) / (2 \log e) = \Theta(\log(\eps \cdot n))$, and obtain that labeling trees of size $n$ for $(1+\eps)$-approximate distances
requires $\Omega(\log(1/\eps)\cdot \log (\eps\cdot n))$ bits. Now,
if $\eps > 1 / \sqrt{n}$ this is in fact $\Omega(\log(1/\eps)\cdot \log n)$ and
we are done. Otherwise ($\eps \leq 1/\sqrt{n}$), we observe that
a scheme with such small $\eps$ can be used for labeling a tree of size $\sqrt{n}$
for exact distances (by subdividing every edge into $\sqrt{n}$ edges).
Such labeling requires $\Omega(\log^2(\sqrt{n}))=\Omega(\log^2n)$ bits,
which for $\eps \geq 1/n$ is also $\Omega(\log(1/\eps)\cdot \log n)$ as required.

\subsection{Upper bound}

We now describe a matching upper bound: a $(1+\eps)$-approximate distance labeling scheme with label size $\Oh(\log(1/\eps)\cdot\log n)$. Our scheme is based on the scheme of Alstrup et al.~\cite{alstrup2015distance}  whose label size is $\Oh(1/\eps\cdot\log n)$. For any node $v$, let $v_1, \dots, v_k$ be the significant ancestors of $v$ in the order they appear on the $v$-to-root path. 
Let $\ceil{x}_{1+\eps}$ denote the smallest power of $1+\eps$ larger than $x$. Observe that $\ceil{x}_{1+\eps}$ is a $(1+\eps)$-approximation of $x$. 

The label of a node $v$ in~\cite{alstrup2015distance} is composed of the following fields:
\begin{enumerate}
\item \label{list:root distance}
$d(v,\treeroot(T))$,
\item \label{list:nca label}
the $\Oh(\log n)$ label generated by~\cref{lem:NCA scheme} applied on $v$,
\item \label{list:significant ancestor sequence}
the sequence $\ceil{\distance(v,v_1)}_{1+\eps}, \ceil{\distance(v,v_2)}_{1+\eps}, \dots, \ceil{\distance(v,v_k)}_{1+\eps}.$
\end{enumerate}
\noindent Let $w = \NCA(u,v)$. If $w = v$ or $w=u$, we can extract the exact distance from~(\ref{list:root distance}). Otherwise, w.l.o.g. we can find the significant ancestor $v_j$ of $v$ such that $v_j=w$ using (\ref{list:nca label}), and then find  $\ceil{\distance(v,w)}_{1+\eps}$ using (\ref{list:significant ancestor sequence}). Alstrup et al. show that:
\begin{align*}
 \distance(u,v) \leq \distance(u,\treeroot(T)) - \distance(v,\treeroot(T)) + 2\cdot\ceil{\distance(v,w)}_{1+\eps} \leq 
	(1 + 2\eps)\cdot \distance(u,v).
\end{align*}

	This means we can compute a $(1 + \eps)$-approximation of $\distance(u,v)$ by replacing $\eps$ with $\eps/2$.
The bottleneck for the size of the label is storing the sequence in (\ref{list:significant ancestor sequence}).
In~\cite{alstrup2015distance}, this sequence is stored using a unary encoding of the sequence $\ceil{\distance(v,v_1)}_{1+\eps}, \ceil{\distance(v,v_2)}_{1+\eps}-\ceil{\distance(v,v_1)}_{1+\eps}, \dots, \ceil{\distance(v,v_k)}_{1+\eps}-\ceil{\distance(v,v_{k-1})}_{1+\eps}$ delimited by a single bit between two consecutive values. The maximal length of the path is at most $n$, so such an encoding will require $\log_{1+\eps}{n}$ bits and additional $k\leq\log n$ bits for the delimiters. This means that the final label size is $\Theta(\log_{1+\eps}{n})$, or $\Theta(1/\eps\cdot\log n)$ for small $\eps$. Instead, we store the sequence using \cref{lem:encodingsequence},
which yields a label of size $\Oh(\log(1/\eps)\cdot\log{n})$ bits and a constant query time.

\bibliographystyle{abbrv}
\bibliography{main}

\end{document}